\documentclass[11pt,letterpaper]{article}

\usepackage[utf8]{inputenc}

\usepackage{amsmath,amssymb,amsfonts} 
\usepackage{epsfig} \usepackage{latexsym,nicefrac,bbm}
\usepackage{color,fancybox,graphicx,subfigure}
\usepackage[top=1in, bottom=1in, left=1in, right=1in]{geometry}
\usepackage{tabularx} \usepackage{hyperref} 
\usepackage[ruled,linesnumbered]{algorithm2e}
\usepackage{enumitem}
\usepackage{multirow}
 \usepackage{url}

 \usepackage{tcolorbox}

\renewcommand{\epsilon}{\varepsilon}

\newcommand{\eps}{\varepsilon}

\newtheorem{theorem}{Theorem}[section]

\newtheorem{definition}{Definition}[section]
\newtheorem{lemma}[theorem]{Lemma}
\newtheorem{remark}[theorem]{Remark}

\newtheorem{corollary}[theorem]{Corollary}

\newenvironment{proof}{\begin{trivlist} \item {\bf Proof:~~}}
   {\qed\end{trivlist}}

\def\FullBox{\hbox{\vrule width 6pt height 6pt depth 0pt}}

\def\qed{\ifmmode\qquad\FullBox\else{\unskip\nobreak\hfil
\penalty50\hskip1em\null\nobreak\hfil\FullBox
\parfillskip=0pt\finalhyphendemerits=0\endgraf}\fi}

\usepackage[T1]{fontenc}
\usepackage{lmodern}

\title{ 
Sampling from Log-Concave Distributions with Infinity-Distance Guarantees\footnote{This is the full version of a paper accepted to NeurIPS 2022 \url{https://openreview.net/pdf?id=nP6e73uxd1}}} 

 \author{Oren Mangoubi\\ Worcester Polytechnic Institute \and Nisheeth K. Vishnoi \\ Yale University}

\usepackage{times}

\begin{document}
\date{}
\maketitle

\begin{abstract}
For a $d$-dimensional log-concave distribution $\pi(\theta) \propto e^{-f(\theta)}$ constrained to a convex body $K$, the problem of outputting samples from a distribution $\nu$ which is $\varepsilon$-close in  {infinity-distance} $\sup_{\theta \in K} |\log  \frac{\nu(\theta)}{\pi(\theta)}|$ to $\pi$ arises in differentially private optimization. While sampling within total-variation distance $\varepsilon$ of $\pi$ can be done by algorithms whose runtime depends polylogarithmically on $\frac{1}{\varepsilon}$, prior algorithms for sampling in $\varepsilon$ infinity distance have runtime bounds that depend polynomially on $\frac{1}{\varepsilon}$. We bridge this gap by presenting an algorithm that outputs a point  $\varepsilon$-close to $\pi$ in infinity distance that requires at most  $\mathrm{poly}(\log \frac{1}{\varepsilon}, d)$ calls to a membership oracle for $K$ and evaluation oracle for $f$, when $f$ is Lipschitz. Our approach departs from prior works that construct  Markov chains on a $\frac{1}{\varepsilon^2}$-discretization of $K$ to achieve a sample with $\varepsilon$ infinity-distance error, and present a method to directly convert continuous samples from $K$ with total-variation bounds to samples with infinity bounds. This approach also allows us to obtain an improvement on the dimension $d$ in the running time for the problem of sampling from a log-concave distribution on polytopes $K$ with infinity distance $\varepsilon$, by plugging in TV-distance running time bounds for the Dikin Walk Markov chain.

\end{abstract}

\newpage

\tableofcontents
\newpage

\section{Introduction} \label{sec_intro}

The problem of sampling from a log-concave distribution is as follows: For a convex body $K \subseteq \mathbb{R}^d$ and a convex function $f:K \rightarrow \mathbb{R}$, output a sample $\theta$ from the distribution 
$\pi(\theta) \propto e^{-f(\theta)}$.
This is a basic problem in computer science, statistics, and machine learning, with applications to optimization and integration \cite{applegate1991sampling, lovasz2006fast}, Bayesian statistics \cite{welling2011bayesian}, reinforcement learning \cite{chapelle2011empirical}, and differential privacy \cite{mcsherry2007mechanism,hardt2010geometry,bassily2014private, leake2020polynomial}.
Sampling exactly from $\pi$ is known to be computationally hard for most interesting cases of $K$ and $f$ \cite{DyerFrieze} and, hence, the goal is to output samples from a distribution $\nu$ that is at a small (specified) ``distance'' to $\pi$.
For applications such as computing the integral of $\pi$, bounds in the total variation (TV) distance \cite{applegate1991sampling} or KL divergence  (which implies a TV bound) are sufficient.
In applications such as computing the expectation of a Lipschitz function with respect to  $\pi$, Wasserstein distance may also be sufficient.
In  differentially private optimization \cite{mcsherry2007mechanism, hardt2010geometry, bassily2014private,ganesh2020faster, leake2020polynomial}, one requires bounds on the stronger infinity-distance -- $$\mathrm{d}_\infty(\nu, \pi):= \sup_{\theta \in K} \left|\log  \frac{\nu(\theta)}{\pi(\theta)}\right|$$ -- to guarantee pure differential privacy, and TV, KL, or Wasserstein  bounds are insufficient; see \cite{dwork2014algorithmic}.

 Pure differential privacy (DP) is the strongest notion of DP and has been extensively studied (see e.g. the survey \cite{dwork2014algorithmic}).
It has advantages  over weaker notions of differential privacy. %
E.g., when privacy of ``groups'' of individuals (rather than just single individuals) must be preserved, 
any mechanism which is  (pure) $\eps$-DP (with respect to single individuals), is also $k\eps$-DP with respect to subsets of $k$ individuals.
Motivated by applications to differential privacy, we study the problem of designing efficient algorithms to output samples from a distribution $\nu$ which is $\varepsilon$-close in $\mathrm{d}_\infty$ to $\pi$.

{\bf Related works.} Several lines of work have designed Markov chains that generate samples from distributions that are close to a given log-concave distribution.
These results differ in both their assumptions on the log-density and its support, as well as the distance used to measure closeness.
One line of work includes bounds for sampling from a log-concave distribution on a compactly supported convex body within TV distance $O(\delta)$, including results with running time that is polylogarithmic in $\frac{1}{\delta}$ \cite{applegate1991sampling, lovasz2007geometry, lovasz2006fast, narayanan2017efficient} (as well as other results which give a running time bound that is polynomial in $\frac{1}{\delta}$ \cite{frieze1994sampling,frieze1999log,bubeck2015finite,brosse2017sampling}).
In addition to assuming access to a value oracle for $f$,  some Markov chains just need access to a membership oracle for  $K$ \cite{applegate1991sampling, lovasz2007geometry, lovasz2006fast}, while others assume that $K$ is a given polytope:  $\{\theta \in \mathbb{R}^d: A\theta \leq b\}$ \cite{kannan2012random,narayanan2016randomized, sachdeva2016mixing, narayanan2017efficient, lee2018convergence}.
They often also assume that $K$ is contained in a ball of radius $R$ and contains a ball of smaller radius $r$.  
Many of these results assume
that the target log-concave distribution satisfies a ``well-rounded''  condition which says that the variance of the target distribution is $\Theta(d)$  \cite{lovasz2007geometry, lovasz2006fast}, or that it is in isotropic position (all  eigenvalues of its covariance matrix are $\Theta(1)$) \cite{lee2017eldan}; when applied to log-concave distributions that are not well-rounded or isotropic, these results require a ``rounding'' pre-processing procedure to find a linear transformation which makes the target distribution well-rounded or isotropic.
Finally, it is often  assumed that the function $f$ is such that $f$ is $L$-Lipschitz or $\beta$-smooth \cite{narayanan2017efficient}, including works handling the widely-studied special case when $f$ is uniform on $K$ where $L=\beta=0$ (see e.g. \cite{lovasz2003hit, kannan2012random, narayanan2016randomized, sachdeva2016mixing, lee2018convergence, chen2017vaidya, laddha2020strong}).

Another line of work gives sampling algorithms with bounds on the distance to the target density $\pi$ in terms of Wasserstein distance \cite{durmus2017nonasymptotic, dalalyan2020sampling}, KL divergence \cite{wibisono2018sampling, durmus2019analysis}, and Renyi divergence \cite{vempala2019rapid}. 
In contrast to works which assume access to an oracle for the value of $f$, many of these results instead assume access to an oracle for the gradient of $f$ and require the log-density to be $L$-Lipschitz or $\beta$-smooth on all of $\mathbb{R}^d$ (or on, e.g., a cube containing $K$) for some $L, \beta>0$.
However, as noted earlier, bounds in the Wasserstein distance, KL divergence, and $\alpha$-Renyi divergence (for $\alpha<\infty$) also do not imply bounds on the infinity distance, and the running time bounds provided by these works are polynomial in $\frac{1}{\eps}$. (See also Appendix \ref{sec_challenges} for additional discussion and challenges.)

Among prior works that give  algorithms with bounds on $d_\infty$, \cite{hardt2010geometry}  applies the grid walk Markov chain of \cite{applegate1991sampling} to sample from a uniform distribution on a convex body. 
\cite{bassily2014private} extends the approach of \cite{hardt2010geometry} to log-Lipschitz log-concave distributions. 
Unlike the TV-distance case where algorithms whose running time depends {\em logarithmically} on the error are known (e.g., \cite{applegate1991sampling, lovasz2006fast, narayanan2017efficient}), the best available bounds for sampling within $O(\eps)$ infinity-distance \cite{hardt2010geometry,bassily2014private} have runtime that is polynomial in $\frac{1}{\eps}$ and a relatively large polynomial in $d$.

{\bf Our contributions.} 
We present a new approach to output samples, which come with $\mathrm{d}_\infty$ bounds, from a log-concave and log-Lipschitz distribution constrained to a a convex body.
Specifically, when $K:=\{\theta: A\theta \leq b\}$ is a polytope (where one is given $A$ and $b$) our main result (Theorem \ref{thm_infinity_divergence_sampler}) guarantees samples from a distribution that is within $O(\eps)$  error in $\mathrm{d}_\infty$ and whose runtime depends 
{logarithmically} on $\frac{1}{\eps}$ compared to the polynomial dependence of  \cite{bassily2014private}.
Our approach departs from prior works that construct  Markov chains on a $\frac{1}{\varepsilon^2}$-discretization of $K$ to achieve a sample with $\varepsilon$ infinity-distance error, and we present a method (Algorithm \ref{alg_TV_to_pure}) to directly convert continuous samples from $K$ with total-variation bounds to samples with infinity bounds (Theorem \ref{thm_TV_to_inf_divergence}). 
This continuous-space approach also allows us to obtain an improvement on the dimension $d$ in the running time when $K$ is a polytope by plugging in TV-distance running time bounds for the Dikin Walk Markov chain of \cite{narayanan2017efficient}.
As  immediate applications, we obtain faster algorithms for differentially private empirical risk minimization (Corollary \ref{corr_DP}) and low rank approximation (Corollary \ref{corrolary_low_rank}).
\section{Results}

Let $B(v,s) := \{z \in \mathbb{R}^d : \|z-v\|_2 \leq s\}$ and
$\omega$ denote the matrix-multiplication constant.

\begin{theorem}[Main result]\label{thm_infinity_divergence_sampler}

There exists an algorithm which, given $\epsilon, L, r, R>0$,  $A \in \mathbb{R}^{m \times d}$, $b\in \mathbb{R}^m$ that define a polytope $K := \{\theta \in \mathbb{R}^d : A \theta \leq b\}$  contained in a ball of radius $R$, a point $a \in \mathbb{R}^d$ such that $K$ contains a ball $B(0,r)$ of smaller radius $r$, and an oracle for the value of a convex function $f: K \rightarrow \mathbb{R}^d$, where $f$ is  $L$-Lipschitz, and defining $\pi$ to be the distribution $\pi \propto e^{-f}$, outputs a point from a distribution $\nu$ such that $\mathrm{d}_\infty(\nu, \pi)< \eps$.
        Moreover, with very high probability\footnote{The number of steps is $O(\tau \times T)$,  where $\mathbb{E}[\tau] \leq 3$, $\mathbb{P}(\tau \geq t) \leq \left(\frac{2}{3}\right)^t$ for $t \geq 0$, and $\tau \leq O(d\log(\frac{R}{r}) + LR)$ w.p. 1.}, this algorithm  takes $O(T)$ function evaluations and $O(T\times md^{\omega-1})$ arithmetic operations, where $T=O((m^2d^3 + m^2 d L^2 R^2) \times [LR + d\log(\frac{Rd +LRd}{r \eps})])$. 

\end{theorem}

\noindent
In comparison to the polynomial in $\frac{1}{\eps}$ runtime bounds of \cite{bassily2014private}, Theorem \ref{thm_infinity_divergence_sampler} guarantees a runtime that is {\em logarithmic} in $\frac{1}{\eps}$, and also improves the dependence on the dimension $d$, in the setting where $K$ is a polytope.
Specifically,  \cite{bassily2014private} show that the number of steps of the grid walk to sample from $\pi$ with infinity-distance error at most $\eps$ is $$O\left(\frac{1}{\eps^2}(d^{10} + d^6L^4R^4)\times \mathrm{polylog}\left( \frac{1}{\eps},\frac{1}{r},R,L,d\right)\right)$$  (Lemma 6.5 in the Arxiv version of \cite{bassily2014private}). 
When applying their algorithm to the setting where $f$ is constrained to a polytope $K =\{x \in \mathbb{R}^d: Ax \leq b\}$,
each step of their grid walk Markov chain requires computing a membership oracle for $K$ and the value of the function $f$.
The membership oracle can be computed in $O(md)$ arithmetic operations.
 Thus the bound on the number of arithmetic operations for each step of their grid walk is $O(md)$ (provided that each function evaluation takes at most $O(md)$ arithmetic operations).
 Thus the bound on the number of arithmetic operations to obtain a sample from $\pi$ is $O(\frac{1}{\eps^2}(md^{11} + md^7L^4R^4)\times \mathrm{polylog}(\frac{1}{\eps},\frac{1}{r}, R,L,d))$.
Thus, Theorem \ref{thm_infinity_divergence_sampler} improves on this bound by a factor of roughly $\frac{1}{\epsilon^2 m^3}d^{8-\omega}$.  For example, when $m=O(d)$, as may be the case in differentially private applications, the improvement is   $\frac{1}{\epsilon^2}d^{5-\omega}$.

We note that the bounds of \cite{bassily2014private} also apply in the more general setting where $K$ is a convex body with membership oracle. 
 One can extend our bounds to achieve a runtime that is logarithmic in $\frac{1}{\epsilon}$ (and polynomial in $d,L,R$) in the more general setting where $K$ is a convex body with membership oracle; we omit the details (see Remark \ref{rem_membership}).

Moreover, we also note that 
 while there are several results which achieve $O(\delta)$ TV bounds in time logarithmic in $\frac{1}{\delta}$, TV bounds do not in general imply $O(\epsilon)$ bounds on the KL or Renyi divergence, or on the infinity-distance, for any $\delta>0$.\footnote{For instance, if $\pi(\theta) = 1$ with support on $[0,1]$, for every $\delta>0$ there is a distribution $\nu$ where $\|\nu-\pi\|_{\mathrm{TV}} \leq 2\delta$ and yet $\mathrm{d}_{\infty}(\nu, \pi) \geq D_{\mathrm{KL}}(\nu,\pi) \geq \frac{1}{2}$.  ($\nu(\theta) = e^{\frac{1}{\delta}}$ on $\theta \in [0,\delta e^{-\frac{1}{\delta}}]$,  $\nu(\theta) = \frac{1-\delta}{1-\delta e^{-\frac{1}{\delta}}}$ on $(\delta e^{-\frac{1}{\delta}},1]$ and $\nu(\theta) = 0$ otherwise)}
 On the other hand, an $\epsilon$-infinity-distance bound does immediately imply a bound of $\epsilon$ on the KL divergence $D_\mathrm{KL}$, and $\alpha$-Renyi divergence $D_\alpha$, since  $D_{\mathrm{KL}}(\mu,\pi) \leq \mathrm{d}_{\infty}(\mu,\pi)$ and $D_\alpha(\mu,\pi) \leq \mathrm{d}_{\infty}(\mu,\pi)$ for any $\alpha>0$ and any pair of distributions $\mu,\pi$. 
 Thus, under the same assumptions on $K$ and $f$, Theorem \ref{thm_infinity_divergence_sampler}  implies a method of sampling from a Lipschitz concave log-density on $K$ with $\epsilon$ KL and Renyi divergence error in a number of arithmetic operations that is {\em logarithmic} in $\frac{1}{\epsilon}$, with the same bound on the number of arithmetic operations.

The polynomial dependence on $\frac{1}{\eps}$ in   \cite{bassily2014private} is due to the fact that they rely on a discrete-space Markov chain \cite{applegate1991sampling}, on a grid with cells of width $w = O(\frac{\eps}{L \sqrt{d}})$, to sample from $\pi$ within $O(\eps)$  infinity-distance.
Since their Markov chain's runtime bound is polynomial in $w^{-1}$, they get a runtime bound for sampling within $O(\eps)$ infinity-distance that is polynomial in $\frac{1}{\eps}$.
The proof of Theorem \ref{thm_infinity_divergence_sampler} bypasses the use of discrete grid-based Markov chains by introducing Algorithm \ref{alg_TV_to_pure} which transforms any sample within  $\delta = O(\eps e^{-d-nLR})$ TV distance of the distribution $\pi \propto e^{-f}$ on the {\em continuous} set $K$ (as opposed to a discretization of $K$), into a sample within $O(\eps)$ infinity-distance from $\pi$.
This allows us to make use of a continuous-space Markov chain, whose step size is not restricted to a grid of width $O(\frac{\eps}{L \sqrt{d}})$ and is instead independent of $\eps$, to obtain a sample within $O(\eps)$ infinity-distance from $\pi$ in time that is {logarithmic} in $\frac{1}{\eps}$.

\begin{theorem} [Main technical contribution: Converting TV bounds to infinity-distance bounds]\label{thm_TV_to_inf_divergence}
There exists an algorithm (Algorithm \ref{alg_TV_to_pure}) which, given $\epsilon,r,R,L >0$, a membership oracle for a convex body $K$ contained in a ball of radius $R$ and containing a ball $B(0,r)$, and an oracle which outputs a point from a distribution $\mu$ which has TV distance $$\delta \leq O\left(\epsilon\times \left(\frac{R(d \log(\nicefrac{R}{r})+LR)^2}{\eps r}\right)^{-d} e^{-LR}\right)$$  from a distribution $\pi \propto e^{-f}$ where $f: K \rightarrow \mathbb{R}$ is an   $L$-Lipschitz function (see Appendix \ref{sec_TV_proof} for the exact values of $\delta$ and related  hyper-parameters), 
outputs a point $\hat{\theta} \in K$ such that the distribution $\nu$ of $\hat{\theta}$ satisfies $\mathrm{d}_{\infty}(\nu, \pi) \leq \epsilon$.
Moreover, with very high probability\footnote{Algorithm \ref{alg_TV_to_pure} finishes in  $\tau$ calls to the sampling and membership oracles, plus $O(\tau d)$ arithmetic operations, where $\mathbb{E}[\tau] \leq 3$ and $\mathbb{P}(\tau \geq t) \leq \left(\frac{2}{3}\right)^t$ for all $t \geq 0$ and $\tau \leq  5d \log(\frac{R}{r}) + 5 LR + 2$ w.p. 1.}, this algorithm finishes in $O(1)$ calls to the sampling and membership oracles, plus $O(d)$ arithmetic operations.
        
\end{theorem}

\noindent
To the best of our knowledge Theorem \ref{thm_TV_to_inf_divergence} is the first result which for any $\eps, L, r,R>0$, when provided as input a sample from a continuous-space distribution on a convex body $K$ within some TV distance $\delta = \delta(\eps,r,R,L)>0$ from a given $L$-log-Lipschitz distribution $\pi$ on $K$, where $K$ is contained in a ball of radius $R$ and containing a ball of smaller radius $r$, outputs a sample with distribution within infinity-distance $O(\eps)$ from $\pi$.
This is in contrast to previous works \cite{bassily2014private} (see also \cite{hardt2010geometry} which applies only to the special case where $\pi$ is the uniform distribution on $K$) which instead require as input a sample with bounded TV distance from the restriction of $\pi$ on a {\em discrete} grid on $K$, and then convert this discrete-space sample into a sample within infinity-distance $O(\eps)$ from the continuous-space distribution $\pi: K \rightarrow \mathbb{R}$.

\begin{algorithm}[H]
\caption{Interior point TV to infinity-distance converter} \label{alg_TV_to_pure}
\KwIn{$d\in \mathbb{N}$}
\KwIn{A membership oracle for a convex body $ K \in \mathbb{R}^d$ and an $r>0$ such that  $B(0,r) \subseteq K$.}
\KwIn{A sampling oracle which outputs a point from a distribution $\mu:K \rightarrow \mathbb{R}$}%

 \KwOut{A point $\hat{\theta} \in K$.} 
  
\textbf{Hyperparameters:} $\Delta >0$, $\tau_{\mathrm{max}}\in \mathbb{N}$ (set in Appendix \ref{sec_TV_proof})

\For{$i = 1, \ldots, \tau_{\mathrm{max}}$}{

Sample a point $\theta \sim \mu$ \label{Line_sampling_oracle}

Sample a point $\xi \sim \mathrm{Unif}(B(0,1))$ \label{Line_Gaussian_sampling}

Set $Z \leftarrow \theta + \Delta r \xi$

Set $\hat{\theta} \leftarrow \frac{1}{1 - \Delta}Z$

If $\hat{\theta} \in K$, output $\hat{\theta}$ with probability $\frac{1}{2}$ and halt.  Otherwise, continue.} \label{Line_membership_oracle}

Sample a point $\hat{\theta} \sim \mathrm{Unif}(B(0,r))$

Output $\hat{\theta}$ 

\end{algorithm}

\begin{remark}[Extension to convex bodies with membership oracles] \label{rem_membership}
We note that Theorem \ref{thm_infinity_divergence_sampler} can be extended to the  general setting where $K$ is an arbitrary convex body in a ball of radius $R$ and containing a ball of smaller radius $r$, and we only have membership oracle access to $K$.
Namely, one can plug in the results of \cite{lovasz2006fast} to our Theorem \ref{thm_TV_to_inf_divergence} to generate a sample from a $L$-Lipschitz concave log-density on an arbitrary convex body $K$ in a number of operations that is (poly)-logarithmic in $\frac{1}{\epsilon}, \frac{1}{r}$ and polynomial on $d, L, R$.
We omit the details.
\end{remark}

\paragraph{Applications to differentially private optimization.}
Sampling from distributions with $O(\epsilon)$  infinity-distance error has many applications to differential privacy.
Here, the goal is to find a randomized mechanism $h: \mathcal{D}^n \rightarrow \mathcal{R}$ which, given a dataset $x \in \mathcal{D}^n$ consisting of $n$ datapoints, outputs model parameters $\hat{\theta} \in \mathcal{R}$ in some parameter space $\mathcal{R}$, which minimize a given (negative) utility function $f(\theta, x)$,
under the constraint that the output $\hat{\theta}$ preserves the pure $\epsilon$-differential privacy of the data points $x$.
A randomized mechanism $h: \mathcal{D}^n \rightarrow \mathcal{R}$ is  said to be  $\epsilon$-differentially private if for any datasets $x, x'  \in \mathcal{D}$ which differ by a single datapoint, and any $S \subseteq \mathcal{R}$, we have that $$\mathbb{P}(h(x) \in S) \leq e^{\epsilon} \mathbb{P}(h(x') \in S);$$ see \cite{dwork2014algorithmic}.

As one application  of Theorem \ref{thm_infinity_divergence_sampler}, we consider the problem of finding an (approximate) minimum $\hat{\theta}$ of an empirical risk function $f:  K  \times \mathcal{D}^n \rightarrow \mathbb{R}$ under the constraint that the output $\hat{\theta}$ is $\epsilon$-differentially private,   where  $f(\theta, x) := \sum_{i=1}^n \ell_i(\theta,x_i)$.
 Following \cite{bassily2014private}, we assume that the $\ell_i(\cdot, x)$ are $L$-Lipschitz for all $x\in \mathcal{D}^n$, $i \in \mathbb{N}$, for some given $L>0$.
In this setting \cite{bassily2014private} show that the minimum ERM utility bound under the constraint that $\hat{\theta}$ is pure $\epsilon$-differentially private,   $\mathbb{E}_{\hat{\theta}}[f(\hat{\theta},x)] - \min_{\theta \in K} f(\theta,x) = \Theta(\frac{d L R}{\epsilon})$,
 is achieved if one samples $\hat{\theta}$ from the exponential mechanism $\pi \propto e^{- \frac{\eps}{2LR}f}$ with infinity-distance error at most $O(\eps)$.
Plugging Theorem \ref{thm_infinity_divergence_sampler} into the framework of the exponential mechanism, we obtain a pure $\eps$-differentially private mechanism which achieves the minimum expected risk (Corollary \ref{corr_DP}, see Section \ref{Sec_DP_ERM} for a proof).

\begin{corollary}[Differentially private empirical risk minimization]\label{corr_DP}
There exists an\, algorithm which, given $\epsilon, L, r, R>0$,  $A \in \mathbb{R}^{m \times d}$, $b\in \mathbb{R}^m$ that define a polytope $K := \{\theta \in \mathbb{R}^d : A \theta \leq b\}$  contained in a ball of radius $R$ and containing a ball $B(0,r)$ of smaller radius $r$,
and a convex function $f(\theta, x) := \sum_{i=1}^n \ell_i(\theta,x_i)$, where each $\ell_i: K \rightarrow \mathbb{R}$ is $L$-Lipschitz,
outputs a random point $\hat{\theta} \in K$ which is pure $\eps$-differentially private and satisfies $$\mathbb{E}_{\hat{\theta}}[f(\hat{\theta},x)] - \min_{\theta \in K} f(\theta,x)\leq O\left(\frac{d L R}{\epsilon}\right).$$
Moreover, this algorithm takes at most   $T \times md^{\omega-1}$ arithmetic operations plus $T$ evaluations of the function $f$, where $T = O\left((m^2d^3 + m^2d n^2 \eps^2\right) \times (\eps n + d) \mathrm{log}^2(\frac{nRd}{r\eps}))$.
\end{corollary}
Corollary \ref{corr_DP} improves on the previous  bound \cite{bassily2014private} of $O((\frac{1}{\epsilon^2}(m+n)d^{11}+ \epsilon^2 n^4 (m + n) d^7)  \times \mathrm{polylog}(\frac{nRd}{r\eps})))$ arithmetic operations by a factor of roughly $\max\left(\frac{d^{8-\omega}}{\epsilon^2 m^2},  \frac{1}{\epsilon m^2}n d^{5}\right)$, in the setting where the  $\ell_i$ are $L$-Lipschitz on a polytope $K$ and each $\ell_i$ can be evaluated in $O(d)$ operations. See Appendix \ref{Sec_DP_ERM} for a proof of this corollary.

As another application, we consider the problem of finding a low-rank approximation of a sample covariance matrix $\Sigma = \sum_{i=1}^n u_i u_i^\top$ where the datapoints $u_i \in \mathbb{R}^d$ satisfy $\|u_i\| \leq 1$, in a differentially private manner. 
Given any $k>0$, the goal is to find a (random) rank-$k$ projection matrix $P$ which maximizes the average variance $\mathbb{E}_P[\langle \Sigma, P \rangle]$ of the matrix $\Sigma$ (also reffered to as the utility of $P$), under the constraint that the mechanism which outputs the matrix $P$ is $\varepsilon$-differentially private.
This problem has many applications to statistics and machine learning, including differentially private principal component analysis (PCA) \cite{chaudhuri2012near, blum2005practical, dwork2014analyze, leake2020polynomial}. 

When privacy is not a concern, the solution $P$ which maximizes the variance is just a projection matrix onto the subspace spanned by top-$k$ eigenvectors of $\Sigma$, and the maximum variance satisfies $\langle \Sigma, P \rangle = \sum_{i=1}^k \lambda_i$, where $\lambda_1 \geq \cdots \geq \lambda_d >0$ denote the eigenvalues of $\Sigma$.
However, when privacy is a concern, there is a tradeoff between  the desired privacy level $\varepsilon$ and the utility $\mathbb{E}_P[\langle \Sigma, P \rangle]$, and the maximum utility $\mathbb{E}_P[\langle \Sigma, P \rangle]$ one can achieve decreases with the privacy parameter $\varepsilon$.
The best current utility bound for an $\varepsilon$-differentially private low rank approximation algorithm was achieved in \cite{leake2020polynomial}, who show that one can find a pure $\epsilon$-differentially private random rank-$k$ projection matrix $P$ such that $\mathbb{E}_P[\langle \Sigma, P \rangle] \geq (1-\delta) \sum_{i=1}^k \lambda_i$ whenever  $\sum_{i=1}^k \lambda_i \geq \Omega\left(\frac{dk}{\epsilon \delta}  \log \frac{1}{\delta}\right)$ for any $\delta>0$.
To generate the matrix $P$, their algorithm generates a sample, with infinity-distance error $O(\epsilon)$, from a Lipschitz concave log-density on a polytope, and transforms this sample into a projection matrix $P$.
The sampling algorithm used in \cite{leake2020polynomial} has a bound of $\mathrm{poly}( \frac{1}{\epsilon}, d, \lambda_1 - \lambda_d)$ arithmetic operations and they
%
leave as an open problem  whether this can be improved from a polynomial dependence on $\frac{1}{\epsilon}$ to a logarithmic dependence on $\frac{1}{\epsilon}$.
Corollary \ref{corrolary_low_rank} shows that  a direct application of Theorem \ref{thm_infinity_divergence_sampler} resolves this problem. (See Section \ref{sec_proof_of_low_rank_DP} for a proof.)

\begin{corollary}[Differentially private low rank approximation]\label{corrolary_low_rank}
There exists an algorithm which, given a sample covariance matrix $\Sigma = \sum_{i=1}^n u_i u_i^\top$ for datapoints $u_i \in \mathbb{R}^d$ satisfying $\|u_i\| \leq 1$, its eigenvalues $\lambda_1 \geq \ldots \lambda_d >0$, an integer $k$, and $\epsilon, \delta>0$, outputs a random rank-k symmetric projection matrix $P$ such that $P$ is $\epsilon$-differentially private and satisfies the utility bound $$\mathbb{E}_P[\langle \Sigma, P \rangle] \geq (1-\delta) \sum_{i=1}^k \lambda_i(\Sigma)$$ whenever  $\sum_{i=1}^k \lambda_i(\Sigma) \geq C \frac{dk}{\epsilon \delta} \log \frac{1}{\delta}$ for some universal constant $C>0$.
Moreover the number of arithmetic operations is logarithmic in $\frac{1}{\epsilon}$ and polynomial in $d$ and $\lambda_1 - \lambda_d$.
\end{corollary}

\section{Proof overviews} \label{sec_technical_overview}

Given any $\epsilon$, and a function $f:\mathbb{R}^d \rightarrow \mathbb{R}$, the goal is to sample from a distribution $\pi(\theta) \propto e^{-f(\theta)}$, constrained to a $d$-dimensional convex body $K$ with infinity-distance error at most $O(\epsilon)$ in a number of arithmetic operations that is logarithmic in $\frac{1}{\epsilon}$.
We assume that $K$ is contained in a ball of some radius $R>0$ and contains a ball of some radius $r>0$, and $f$ is  $L$-Lipschitz.
In addition we would also like our bounds to be polylogarithmic in $\frac{1}{r}$, and  polynomial in $d,L, R$ with a lower-order dependence on the dimension $d$ than currently available bounds for sampling from Lipschitz concave log-densities on a polytope in infinity-distance \cite{bassily2014private}.
We note that since whenever $K$ is contained in a ball of radius $R$ and contains a ball $B(0,r)$ of smaller radius $r$, we also have that $B(0,r) \subseteq K \subseteq B(0,2R)$, without loss of generality, we may assume that  $B(0,r) \subseteq K \subseteq B(0,R)$ as this would only change the bounds provided in our main theorems by a constant factor.

The main ingredient in the proof of  Theorem \ref{thm_infinity_divergence_sampler} is  Theorem \ref{thm_TV_to_inf_divergence} that uses Algorithm \ref{alg_TV_to_pure} to transform a TV-bounded sample into a sample from $\pi$ with error bounded in $\mathrm{d}_\infty$.
{Subsequently, we invoke Theorem \ref{thm_infinity_divergence_sampler}  when $K$ is given as a polytope $K:=\{x \in \mathbb{R}^d: Ax \leq b\}$ and 
 plug in the Dikin Walk Markov chain of \cite{narayanan2017efficient} which generates independent samples from $\pi$ with bounded TV error.}
 We first present an overview of the proof of Theorem \ref{thm_TV_to_inf_divergence}. (The full proof has been omitted to space restrictions and presented in Appendix \ref{sec_TV_proof}.)
 The proof of Theorem \ref{thm_infinity_divergence_sampler} is presented in Section \ref{sec_completing_main_result}.

\subsection{Converting samples with TV bounds to infinity-distance bounds; proof of Theorem \ref{thm_TV_to_inf_divergence}} \label{sec_continuous_to_TV_overview}

{\bf Impossibility of obtaining log-dependence on infinity-distance via  grid walk. }
One approach is to observe that  if $e^{-f}$ has support on a discrete space $S$ with at most $|S|$ points, then any $\nu$ such that $\|\nu - \pi\|_{\mathrm{TV}} \leq \epsilon$ also satisfies $$\mathrm{d}_{\infty}(\nu, \pi) \leq 2 |S| \frac{\max_{z\in S} e^{-f(z)}}{\min_{z\in S} e^{-f(z)}} \times \epsilon$$ for any $\epsilon \leq \min_{z\in S} \pi(z)$.
This suggests forming a grid $G$ over $K$, then using a discrete Markov chain to generate a sample $\theta$ within  $O(\epsilon)$ TV distance of the discrete distribution $\pi_G \propto e^{-f}$ with support on the grid $G$, and then designing an algorithm which takes as input $\theta$ and outputs a point with bounded infinity-distance to the continuous distribution $\pi$.
This approach was used in \cite{hardt2010geometry} in the special case when $\pi$ is uniform on $K$, and then extended by \cite{bassily2014private} to log-Lipschitz log-concave distributions.
In their approach, \cite{bassily2014private} first run a ``grid-walk'' Markov chain on a discrete grid in a cube containing $K$.
They then apply the bound from \cite{applegate1991sampling} which says that the grid walk obtains a sample $Z$ within TV distance $O(\delta)$ from the distribution $\propto e^{-f}$ (restricted to the grid) in time that is polylogarithmic in $\frac{1}{\delta}$ and quadratic in $a^{-1}$, where $a$ is the distance between neighboring grid points. 
Since their grid has size $|S| = \Theta((\frac{R}{a})^d)$, a TV distance of $O(\delta)$ automatically implies an infinity-distance of $O(\delta c |S|)$, where $c$ is the ratio of the maximum to the minimum probability mass satisfying $c\leq e^{LR}$ since $f$ is $L$-Lipschitz on $K \subseteq B(0,R)$. 
Thus,  by using the grid walk to sample within TV-distance $\delta = O\left(\frac{\epsilon}{|S|c}\right)$ from the discrete distribution $\pi_G$,  they obtain a sample $Z$ which also has infinity-distance $O(\epsilon)$ from $\pi_G$.
Finally, to obtain a sample from the distribution $\pi \propto e^{-f}$ on the continuous space $K$, 
they sample a point uniformly from the ``grid cell'' $[Z-a, Z+a]^d$ centered at $Z$.
Since $f$ is $L$-Lipschitz, the ratio $\frac{e^{-f(Z)}}{e^{-f(w)}}$ is bounded by $O(\eps)$ for all $w$ in the grid cell $[Z-a, Z+a]^d$ as long as $a= O\left(\frac{\eps}{L \sqrt{d}}\right)$, implying that the sample is an infinity-distance of $O(\eps)$ from $\pi \propto e^{-f}$.
However, since the running time bound of the grid walk is quadratic in $a^{-1}$, the grid coarseness $a= O\left(\frac{\eps}{L \sqrt{d}}\right)$ needed to achieve $O(\eps)$ infinity-distance from $\pi$ leads to a running time bound which is {\em quadratic} in $\frac{1}{\eps}$.

To get around this problem, rather than relying on the use of a discrete-space Markov chain such as the grid walk to sample within $O(\eps)$ infinity-distance from $\pi$, we introduce an algorithm (Algorithm \ref{alg_TV_to_pure}) which transforms any sample within  $\delta = O\left(\eps e^{-d-LR}\right)$ TV distance from the distribution $\pi \propto e^{-f}$ on the {\em continuous} space $K$ (as opposed to a grid-based discretization of $K$), into a sample within $O(\eps)$ infinity-distance from $\pi$.
This allows us to make use of a continuous-space Markov chain, such as the Dikin walk of \cite{narayanan2017efficient}, whose step-size is not restricted by a grid of coarseness $w= O\left(\frac{\eps}{L \sqrt{d}}\right)$ and instead is independent of $\eps$, in order to generate a sample within $O(\eps)$ infinity-distance from $\pi$ in runtime that is {\em logarithmic} in $\frac{1}{\eps}$.

\smallskip
\noindent
{\bf Converting continuous space TV-bounded samples  to infinity-distance bounded samples. }
As discussed in Section \ref{sec_intro}, there are many Markov chain results which allow one to sample from a log-concave distribution on $K$ with error bounded in weaker metrics such as total variation, Wasserstein, or KL divergence.
However, when sampling from a continuous distribution, bounds in these  metrics do not directly imply bounds in infinity-distance.
And techniques used to prove bounds in weaker metrics do not easily extend to methods for bounding the infinity-distance; see Section \ref{sec_challenges_spectral_methods}.

\smallskip
\noindent
{\em Convolving with continuous noise.} 
As a first attempt, we consider the following simple algorithm: sample a point $\theta \sim \mu$ from a distribution $\mu$ with total variation error $\|\mu - \pi \|_{\mathrm{TV}} \leq O(\epsilon)$.
Since $f$ is $L$-Lipschitz, for any $\Delta< \frac{\epsilon}{L}$ and any ball $B(z,\Delta)$ in the $\Delta$-interior of $K$ (denoted by $\mathrm{int}_\Delta(K)$; see Definition \ref{def:interior}), we can obtain a sample from a distribution $\nu$ such that $\log\left(\frac{\nu(z)}{\mu(z)}\right) \leq \epsilon$ for all $z \in \mathrm{int}_\Delta(K)$ by convolving $\mu$ with the uniform distribution on the ball $B(0,\Delta)$.
Sampling from this distribution $\nu$ can be achieved by first sampling $\theta \sim \mu$ and then adding noise $\xi \sim \mathrm{Unif}(B(0, \Delta))$ to the sample $\theta$.

Unfortunately, this simple algorithm does not allow us to guarantee that $\log(\frac{\nu(z)}{\mu(z)}) \leq \epsilon$ at points $z \notin  \mathrm{int}_\Delta(K)$ which are a distance less than $\Delta$ from the boundary of $K$.
To see why, suppose that $K = [0,1]^d$ is the unit cube, that $f$ is constant on $K$, and consider a point $w = (1,\ldots, 1)$ at the corner of the cube $K$.
In this case we could have that $\nu(z) \leq 2^{-d} \pi(z)$ for all $z$ in some ball containing $w$, and hence $\mathrm{d}_{\infty}(\nu, \pi) = \sup \left| \log(\frac{\nu(z)}{\pi(z)})\right| \geq d \log(2)$, no matter how small we make $\Delta$.

\smallskip
\noindent {\em Stretching the convex body to handle points close to the boundary.}
To get around this problem, we would like to design an algorithm which samples from some distribution $\nu$ such that $\left|\log \frac{\nu(z)}{\pi(z)}\right| \leq \epsilon$ for all $z \in K$, including at points $z$ near the corners of $K$.
Towards this end, we first consider the special case where $K$ is itself contained in the $\Delta$-interior of another convex body $K'$, the function $f: K \rightarrow \mathbb{R}$ extends to an $L$-Lipschitz function on $K'$ (also referred to here with slight abuse of notation as $f$) %
and where we are able to sample from the distribution $\propto e^{-f}$ on $K'$ with $O(\epsilon)$ total variation error.
If we sample $\theta \sim e^{-f}$ on $K'$ with total variation error $O(\delta)$ where $\delta \leq \epsilon e^{-d \log(R)}$, add noise $\xi \sim \mathrm{Unif}(B(0,\Delta))$ to $\theta$ for $\Delta = \frac{\delta}{LR}$, and then reject $\theta + \xi$ only if it is not in $K$, we obtain a sample whose distribution is $O(\epsilon)$ from the distribution $\propto e^{-f}$ on $K$ in infinity-distance.

However, we would still need to define and sample from such a convex body $K'$, and to make sure that $K'$ is not too large when compared to $K$; otherwise the samples from the distribution $\propto e^{-f}$ on $K'$ may be rejected with high probability.
Moreover, another issue we need to deal with is that $f$ may not even be defined outside of $K$.

\begin{figure}[h]
    \centering
    \includegraphics[width=0.6\textwidth]{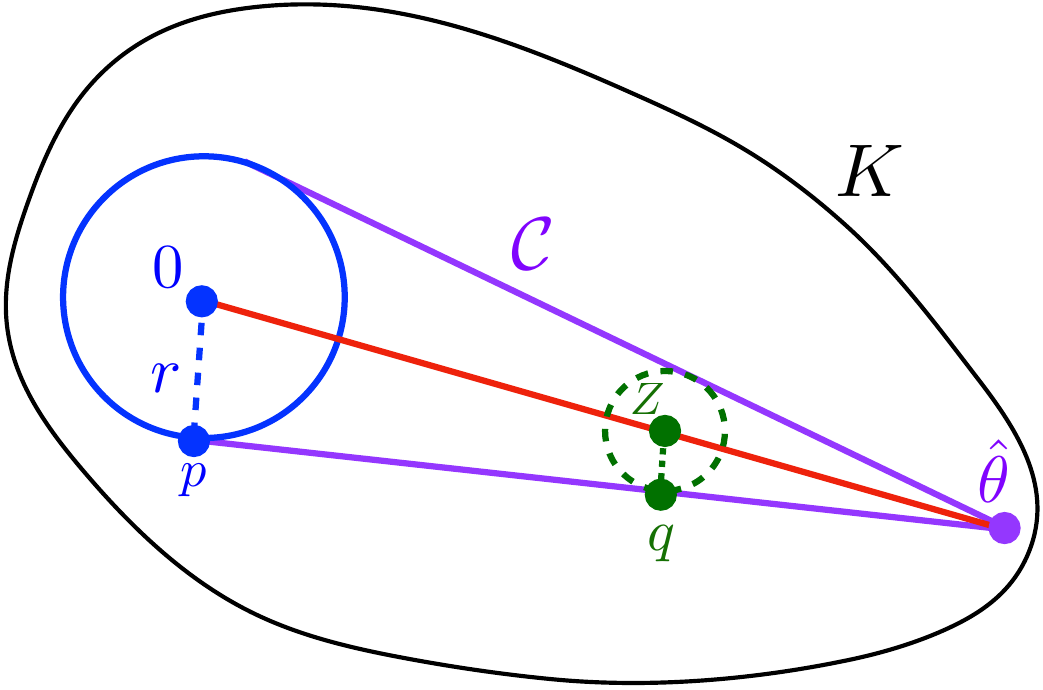}
    \caption{The construction used in the proof of Lemma \ref{Lemma_cvx_hull}.}
    \label{convex_hull_diagram}
\end{figure}

\noindent
To get around these two problems, in Algorithm \ref{alg_TV_to_pure}, we begin by taking as input a point $\theta \sim \mu$ sampled from some distribution $\mu$ supported on $K$ where $\|\mu- \pi\|_{\mathrm{TV}} \leq \delta$ for some $\delta \leq \epsilon e^{-d \log(R)}$, and add noise $\xi \sim \mathrm{unif}(B(0, \Delta r))$ in order to sample from a distribution $\hat{\mu}$ which satisfies $\left|\log \frac{\hat{\mu}(z)}{\pi(z)}\right| \leq \epsilon$ for all $z  \in \mathrm{int}_{\Delta r}(K)$. 
Here $r$ is the radius of the small ball contained in $K$; the choice of radius $\Delta r$ for the noise distribution is because we will show in the following paragraphs that to sample from the distribution $\pi$ on $K$ with infinity-distance error $\eps$ it is sufficient sample a point in the $\Delta r$-interior of $K$ and to then apply a ``stretching'' operation to $K$.

We still need a method of sampling  within $O(\epsilon)$ infinity-distance error of $\pi$ on all of $K$, including in the region $K\backslash \mathrm{int}_{\Delta r}(K)$ near the boundary of $K$.
Towards this end, after Algorithm \ref{alg_TV_to_pure} generates a point $Z = \theta + \xi$ from the above-mentioned distribution $\hat{\mu}$,
it then multiplies $Z$ by $\frac{1}{1-\Delta}$ and returns the resulting point $\hat{\theta}:= \frac{1}{1-\Delta}Z$ if it is $K$, in other words, if $Z \in (1-\Delta)K$.
If we can show that $(1-\Delta) K \subseteq \mathrm{int}_{\Delta r}(K)$, then this would imply that $\left|\log \frac{\hat{\mu}(z)}{\pi(z)}\right| \leq \epsilon$ for all $z  \in (1-\Delta) K$, and hence that the distribution of $\hat{\nu}$ of the returned point $\hat{\theta}$ satisfies $$    \left|\log \frac{\hat{\nu}(z)}{\pi((1-\Delta) z)}\right| \leq \left|\log \frac{\hat{\mu}((1-\Delta)z)}{\pi((1-\Delta)z)}\right| + \log \frac{1}{(1-\Delta)^d}\leq O(\epsilon)$$ for all $z \in K$.
Since $f$ is $L$-Lipschitz we have $\left|\log\frac{\pi(\theta)}{\pi((1-\Delta)\theta)}\right|=O(\epsilon)$ for all $\theta \in K$, and hence we would then have that the distribution $\hat{\nu}$ of the point $\hat{\theta}$ returned by Algorithm \ref{alg_TV_to_pure} satisfies
\begin{equation} \label{eq_TO_1}
    \left|\log \frac{\hat{\nu}(z)}{\pi(z)}\right| = \left|\log \frac{\hat{\nu}(z)}{\pi((1-\Delta) z)}\right| + O(\epsilon) \leq O(\epsilon) \qquad \forall z \in K.
\end{equation}
However, for \eqref{eq_TO_1} to hold, we still need to show that $(1-\Delta)K \subseteq \mathrm{int}_{\Delta r}(K)$  (proved in Lemma \ref{Lemma_cvx_hull}).
In other words, we would like to show that for any point $Z \in (1-\Delta)K$, there is a ball $B(Z, \Delta r)$ centered at $Z$ of radius $\Delta r$ contained in $K$.
To show this fact, it is sufficient to consider the convex hull $\mathcal{C}$ of $B(0,r) \cup \left\{\frac{1}{1-\Delta} Z \right\} \subseteq K$, and show that it contains the ball $B(Z, \Delta r)$.
Towards this end, we make the following geometric construction: we let $p$ be a point such that the line $p\hat{\theta}$ is tangent to $B(0,r)$, and $q$ the point on $p\hat{\theta}$ which minimizes the distance $\|q-Z\|_2$ (see Figure \ref{convex_hull_diagram}).
Since $\angle 0 p \hat{\theta}$ and $\angle Z q \hat{\theta}$ are both right angles, we have that the triangles $0p\hat{\theta}$ and $Z q \hat{\theta}$ are similar triangles, and hence that $\frac{\|Z - q\|_2}{r} = \frac{\|Z - \hat{\theta}\|_2}{\|\hat{\theta}-0\|_2}$.  
In other words,
\begin{eqnarray*} \|Z - q\|_2
 = \frac{\|Z - \hat{\theta}\|_2}{\|\hat{\theta}-0\|_2} \times r
 = \frac{\left\|Z - \frac{1}{1-\Delta} Z\right\|_2}{\left\|\frac{1}{1-\Delta} Z\right\|_2} \times r
 =  \Delta \times r,\\
\end{eqnarray*}
implying a ball of radius $\Delta r$ centered at $Z$  is contained in $\mathcal{C}\subseteq K$, and hence that $$(1-\Delta)K \subseteq \mathrm{int}_{\Delta r}(K).$$

\smallskip
\noindent
{\em Bounding the infinity distance error.}
 To complete the bound on the infinity-distance  of the distribution $\hat{\nu}$ of the point returned by Algorithm \ref{alg_TV_to_pure} to the target distribution $\pi$, we must show both a lower bound (Lemma \ref{lemma_lower_bound}) and an upper bound (Lemma \ref{lemma_upper_bound}) on the ratio $\frac{\hat{\nu}(\theta)}{\pi(\theta)}$ at every point $\theta \in K$.
Both the upper and lower bounds are necessary to bound the infinity-distance $\mathrm{d}_{\infty}(\hat{\nu}(\theta), \pi(\theta)) = \sup_{\theta \in K} \left|\log  \frac{\hat{\nu}(\theta)}{\pi(\theta)}\right|$.

Both Lemmas \ref{lemma_lower_bound} and \ref{lemma_upper_bound} require the input point to have TV error $\delta < \eps (\frac{R}{\Delta r})^{-d} e^{-LR}$.    
The term $(\frac{R}{\Delta r})^{-d}$ is a lower bound on the ratio of the volume of $K$ to the volume of the smoothing ball $B(0,\Delta r)$; this bound holds since $K$ is contained in a ball of radius $R$.
The term $e^{-LR}$ is a lower bound on the ratio $\frac{\min_{w \in K} \pi(w)}{\max_{w \in K} \pi(w)}$ of the minimum value of the density  $\pi$ to the maximum value of $\pi$ at any two points in $K$; this bound holds since $f$ is $L$-Lipschitz.

The above choice of $\delta$ ensures that in any ball $B(z,\Delta r)$ with center $z$ in the $\Delta r$-interior of $K$, the distribution $\mu$ of the input point, which satisfies $\|\mu- \pi\|_{\mathrm{TV}} \leq \delta$, will have between $e^{-\eps}$ and $e^{\eps}$ times the probability mass which the target distribution $\pi$ has inside the ball $B(z,\Delta r)$.
Thus, when the distribution $\mu$ is smoothed by adding noise uniformly distributed on a ball of radius $\Delta r$, the smoothed distribution $\tilde{\nu}(\theta)$ is within $e^{-\eps}$ and $e^{\eps}$ times the target probability density $\pi(\theta)$ at any point $\theta$ in the $\Delta r$-interior of $K$, allowing us to bound the infinity distance error of the smoothed distribution $\tilde{\nu}$ at any point $\theta$ in the $\Delta r$-interior of $K$.
We then apply this fact, together with Lemma \ref{Lemma_cvx_hull} which says that for any point $\theta \in K$ the point $(1- \Delta)\theta$ is in the $\Delta r$-interior of $K$, to bound the distribution $\nu$ of the output point (after the stretching operation) as follows,
\begin{equation}\label{eq_overview1}
\nu(\theta) \geq (1- \Delta)^d \tilde{\nu}((1- \Delta)\theta) 
\stackrel{\textrm{Lemma } \ref{Lemma_cvx_hull}}{\geq} (1- \Delta)^d \pi((1- \Delta)\theta) \times e^{-\epsilon}
\geq \pi(\theta) e^{-\frac{\epsilon}{2}}.
\end{equation}
 Here our choice of hyperparameter $\Delta \leq \frac{\eps}{\max(d, LR)}$ ensures that $(1- \Delta)^d = \Omega(1)$ and, since $f$ is $L$-Lipschitz, that $\pi((1- \Delta)\theta) \geq e^{-\epsilon} \pi(\theta)$.
This proves the lower bound (Lemma \ref{lemma_lower_bound}).
The proof of the upper bound (Lemma \ref{lemma_upper_bound}) follows in a similar way as equation \eqref{eq_overview1} but with the inequalities going in the opposite direction.

\smallskip
\noindent
{\em Bounding the number of iterations and concluding the proof of Theorem \ref{thm_TV_to_inf_divergence}.}
We still need to deal with the problem that the point $\hat{\theta}$ may not be accepted.
If this occurs, roughly speaking, we repeat the above procedure until a point $\hat{\theta}$ with distribution $\hat{\nu}$ is accepted.
To bound the number of iterations, we show that $\hat{\theta}$ is in $K$ with high probability.
Towards this end, we first use the facts that $f$ is $L$-Lipschitz and $K \subseteq B(0,R)$, to show that the probability a point sampled from $\pi \propto e^{-f}$ lies inside $(1-\Delta) K$ is at least $(1- \Delta)^d e^{- L \Delta R} \geq \frac{9}{10}$ (Lemma \ref{lemma_rejection_probability}).
Lemma \ref{lemma_rejection_probability} says that if you stretch the polytope by a factor of $\frac{1}{1-\Delta}$, then most of the volume of the stretched polytope $(\frac{1}{1-\Delta})K$ remains inside the original polytope $K$.
The term $(1- \Delta)^d$ is just the ratio of the volume of  $(1-\Delta) K$ to the volume of $K$.
And, since $f$ is $L$-Lipschitz, the term $e^{- L \Delta R}$ bounds the ratio $\frac{\pi(\theta)}{\pi(\frac{1}{1-\Delta}\theta)}$ of the target density at any point $\theta \in K$ to the value of $\pi$ at the point $\frac{1}{1-\Delta}\theta$ to which the stretching operation transports $\theta$, whenever $\frac{1}{1-\Delta}\theta \in K$.
The choice of hyperparameter $\Delta \leq \frac{\eps}{\max(d, LR)}$ ensures that the acceptance probability $(1- \Delta)^d e^{- L \Delta R}$ guaranteed by Lemma \ref{lemma_rejection_probability} is at least  $\frac{9}{10}$.
Since the convex body $(1-\Delta)K$ contains the ball $B(0,\frac{r}{2})$, applying Lemma \ref{Lemma_cvx_hull} a second time (this time to the convex body $(1-\Delta)K$) 
we get that  $$(1-3\Delta) K \subseteq \mathrm{int}_{\Delta r} ((1-\Delta)K).$$
Thus, by Lemma \ref{lemma_rejection_probability} we have that  
 $\theta$ lies inside $\mathrm{int}_{\Delta r} ((1-\Delta)K)$ with probability at least $\frac{9}{10} - \delta \geq \frac{8}{10}$ (as $\theta$ is sampled from $\pi$ with TV error $\leq \delta$).
Therefore, since $\xi \sim B(0,\Delta r)$, we must also have that the probability that the point $\hat{\theta} = \frac{1}{1-\Delta}(\theta + \xi)$ is in $K$ (and is therefore not rejected) is greater than $\frac{8}{10}$ \footnote{In Algorithm \ref{alg_TV_to_pure} we reject $\hat{\theta}$ with a slightly higher probability to ensure that, in differential privacy applications, in addition to the privacy of the point returned by the algorithm, the runtime is also $\epsilon$-differentially private.}.
This implies that the number of iterations until our algorithm returns a point $\hat{\theta}$ is less than $k>0$ with probability at least $1-2^{-k}$, and the expected number of iterations is at most 2 (proved in Corollary \ref{cor_runtime}).

Since each iteration requires one random sample $\theta$ from the distribution $\mu$, and one call to a membership oracle for $K$ (to determine if $\frac{1}{1-\Delta}Z \in K$), the number of sampling oracle and membership oracle calls required by Algorithm \ref{alg_TV_to_pure} is $O(1)$ with very high probability.
Therefore, with high probability, Algorithm \ref{alg_TV_to_pure} returns a point $\hat{\theta}$ from a distribution with infinity-distance at most $\epsilon$ from $\pi$ after $O(1)$ calls to the sampling and membership oracles.

Since Algorithm \ref{alg_TV_to_pure} succeeds with probability $1- 2^{-k}$ after $k$ iterations, after $$\tau_{\mathrm{max}} = 5d \log(\frac{R}{r}) + 5 LR + \eps$$ iterations Algorithm \ref{alg_TV_to_pure} will have succeeded with probability roughly $1-\eps(\frac{R}{r})^{-5d} e^{-5LR}$.
In the very unlikely event that Algorithm \ref{alg_TV_to_pure} still has not succeeded after $\tau_{\mathrm{max}}$ iterations, Algorithm \ref{alg_TV_to_pure}  simply outputs a point sampled from the uniform distribution on the ball $B(0,r)$ of radius $r$ contained in $K$.
The probability mass of the target distribution $\pi$ inside this ball is at least as large as $(\frac{R}{r})^{-d} e^{-LR}$; thus, since $f$ is $L$-Lipschitz, we show in Corollary \ref{cor_runtime} that outputing a sample from the uniform distribution on this ball with probability $\eps(\frac{R}{r})^{-5d} e^{-5LR}$ does not change the $\infty$-distance error of the sample returned by the algorithm by more than $\eps$.

\subsection{Completing the proof of Theorem \ref{thm_infinity_divergence_sampler}} \label{sec_completing_main_result}

\begin{proof}
By Theorem \ref{alg_TV_to_pure}, the output of Algorithm \ref{alg_TV_to_pure} has infinity-distance error bounded by $\epsilon$ as long as the input samples have TV distance error bounded by $$\delta \leq O\left(\epsilon\times \left(\frac{R(d \log(\nicefrac{R}{r})+LR)^2}{\eps r}\right)^{-d} e^{-LR}\right),$$ and, with high probability, Algorithm \ref{alg_TV_to_pure} requires $O(1)$ such independent samples.
To generate a sample from $\pi$ with TV error $O(\delta)$ when $K = \{\theta \in \mathbb{R}^d : A \theta \leq b\}$ is a polytope defined by $m$ inequalities, we use the Dikin Walk Markov chain of \cite{narayanan2017efficient}.
This Markov chain requires an initial point from some distribution $\mu_0$ which is $w$-warm with respect to the stationary distribution $\pi$, that is, $\sup_{z\in K} \frac{\mu_0(z)}{\pi(z)} \leq w$.
To obtain a warm start, we let $\mu_0$ be the uniform distribution on the ball with radius $r$ contained in $K$ and sample from $\mu_0$.
Since $f$ is $L$-Lipschitz, and $K$ is contained in a ball of radius $R$,  $\mu_0$  is $w$-warm with $$ w \leq \frac{1}{\mathrm{Vol}(B(0,r))} \times  \left(  \frac{\max_{z \in K} \pi(\theta)}{\min_{z \in K} \pi(\theta)} \times \mathrm{Vol}(B(0, R)) \right ) \leq \left(\frac{R}{r}\right)^d \times e^{RL}.$$
\noindent
From \cite{narayanan2017efficient}, we have that from this $w$-warm start the Dikin Walk Markov chain requires at most $O((m^2d^{4} + m^2d^{2}L^2R^2)\log(\frac{w}{\delta}))$ steps to generate a sample with TV distance at most $\delta$ from $\pi$, where each step makes one function evaluation and $O(md^{\omega-1})$ arithmetic operations.
Plugging in the above values of $\delta, w$ the number of Markov chain steps is $$T=O((m^2d^3 + m^2 d L^2 R^2) \times [LR + d\log(\frac{Rd +LRd}{r \eps})])$$ to generate each independent sample with the required TV error $O(\delta)$.
Since the number of independent samples required as input for Algorithm \ref{alg_TV_to_pure} is $O(1)$ w.h.p., the number of arithmetic operations for Algorithm \ref{alg_TV_to_pure} to output a point with at most $\epsilon$ infinity-distance error is $O(T \times md^{\omega-1})$.

Finally, we note that in the more general setting where $K$ is a convex body with membership oracle (but not necessarily) a polytope, we can instead use, for instance, the hit-and-run Markov chain of \cite{lovasz2006fast} to generate samples from $\pi$ with TV error $O(\delta)$ in a number of membership and function evaluation oracle calls that is polynomial in  $d$ and poly-logarithmic in $\frac{1}{\delta}, R,r$.
We can then plug this sample into our Algorithm \ref{alg_TV_to_pure} to obtain a sample from $\pi$ with infinity-distance error $O(\eps)$ in a number of oracle calls that is (poly)-logarithmic in $\frac{1}{\epsilon}, \frac{1}{r}$ and polynomial on $d, L, R$. (see Remark \ref{rem_membership}). 
\end{proof}

\section{Challenges obtaining infinity-distance bounds from continuous-space Markov chains} \label{sec_challenges}

\subsection{Challenges in obtaining infinity-distance bounds via spectral gap methods} \label{sec_challenges_spectral_methods}
 Many TV bounds for Markov chains have been obtained by applying isoperimetric inequalities for log-concave distributions $\pi$, to bound the spectral gap $\gamma := 1- \lambda_2(\mathcal{K})$ of the Markov transition kernel operator $\mathcal{K}$ (see e.g. \cite{lovasz1993random, lovasz2003hit}).
If the initial distribution $\mu_0$ is such that $\sup_{\theta \in K} \frac{\mu_0(\theta)}{\pi(\theta)}$ is bounded by some number $w$ (e.g., by initializing the Markov chain at a uniform random point in a ball contained in the interior of the polytope), a bound on the spectral gap of $\mathcal{K}$ implies $O(\delta)$ bounds on the total variation error in a number of steps that is logarithmic in $\frac{w}{\delta}$. 

Unfortunately, bounding the spectral gap does not in general allow one to obtain bounds on the infinity-distance.
While one can bound the $\chi^2$-divergence (which implies a bound on the TV distance),
by using the fact (first shown in \cite{lovasz1993random}) that $\|\mathcal{K}^t u_0\|_2 \leq  (1- \gamma)^t \|u_0\|_2$,  where $u_0 := \mu_0 - \mathrm{proj}_\pi(\mu_0)$, a bound on the spectral gap does not imply a bound in the infinity-distance error.
The difficulty in bounding the infinity-distance error arises because, if the space $S$ is continuous, even though $\|\mathcal{K}^t u_0\|_2 \leq  (1- \gamma)^t \|u_0\|_2$,
 there may still be some $c>0$ for which $\|\mathcal{K}^t u_0\|_\infty \geq c$  for all $t>0$.
For instance this is the case when $\pi$ is the uniform distribution on a polytope $K$ and $\mathcal{K}$ is the transition Kernel of the Dikin walk,
since for every time $t$, there is always a ball $B_t \subseteq K$ sufficiently close to the boundary of $K$ such that the Dikin walk has probability zero of entering $B_t$ after $t$ steps (for the Gaussian Dikin walk, the probability of entering $B_t$ is very low, but still nonzero). 
Thus, $\mathcal{K}^t \mu_0(z) = 0$ for all points $z \in B_t$, and, since  $\pi(z) = \frac{1}{\mathrm{vol}(K)}$ at every point $z \in K$, we must have  $\|\mathcal{K}^t\mu_0 - \pi\|_\infty =\|\mathcal{K}^t u_0\|_\infty \geq \frac{1}{\mathrm{vol}(K)}$ for all $t \geq 0$.

On the other hand, in the special case when the space $S$ is discrete and has a finite number of elements $|S|$, the ($\ell^2$-normalized) eigenvectors $v$ have bounded infinity norm, $\|v\|_\infty \leq 1$, and, hence $\|\mathcal{K}^t \mu_0 - \mathrm{proj}_\pi(\mu_0) \|_\infty \leq (1-\gamma)^t |S|$. 
Thus, if $S$ is discrete with finitely many elements, bounding the spectral gap implies one can sample from $\pi$ with infinity-distance error $O(\epsilon)$ in a number of steps that is logarithmic in $\frac{1}{\epsilon} \times \frac{|S|}{\min_{i\in[S]} \pi[i]}$ if $\min_{i\in[S]} \pi[i]>0$  (where, with slight abuse of notation, we denote by $\pi$ the probability mass function of the discrete Markov chain's stationary distribution).

Aside from bounding the spectral gap, many works instead make use of probabilistic coupling methods to bound the distance of a continuous-space Markov chain to the target distribution in, e.g., the Wasserstein distance metric (see for instance \cite{durmus2017nonasymptotic, dalalyan2017theoretical}). 
And other works instead achieve bounds in the KL divergence metric for, e.g., the Langevin dynamics Markov chain by analyzing the (continuous-time) Langevin diffusion as a gradient flow of the KL divergence functional in the space of probability distributions 
(see for instance \cite{wibisono2018sampling, cheng2018convergence, durmus2019analysis}). 
In Appendix \ref{sec_challenges_coupling} and \ref{sec_challenges_KL} we discuss challenges which arise if one seeks to extend either of these methods to obtain $O(\epsilon)$ bounds in the infinity-distance metric, in a number of Markov chain steps that is polylogarithmic in $\frac{1}{\epsilon}$.

\subsection{Challenges in extending coupling-based analysis from Wasserstein (and TV) bounds to infinity-distance bounds}\label{sec_challenges_coupling}
As an alternative to bounding the spectral gap of a Markov chain, one can oftentimes instead make use of a probabilistic coupling method to bound the error of the Markov chain.
Here, one considers two Markov chains, one Markov chain started at an arbitrary initial point which is the initial point provided to the Markov chain sampling algorithm, and another ``imaginary'' Markov chain (oftentimes just a copy of the algorithm's Markov chain) which is started at a random point distributed according to the target distribution $\pi$ and for which $\pi$ is a stationary distribution.
The goal is to find a joint distribution-- also called a probabilistic coupling-- for the steps of the two Markov chains such that the distance between the two chains becomes very small in some metric of interest after the Markov chains take multiple steps.
If one can find such a coupling, then one can bound the distance of the algorithm's Markov chain to the target distribution $\pi$ in the relevant metric.
If a coupling is found such that the distance between the two Markov chains contracts in the Euclidean distance, then this implies bounds in the Wasserstein metrics.
A contraction of the expected (squared) Euclidean distance implies bounds in the $1$- (or $2$)-Wasserstein metric (see e.g. \cite{durmus2017nonasymptotic, dalalyan2017theoretical}).
However, for any $k \in \mathbb{N}$, a bound on the $k$-Wasserstein distance $W_k(\nu,\pi) := \sup_{\rho\sim \Pi(\mu, \pi)} (\mathbb{E}_{(X,Y) \sim \rho}[\| X-Y \|^k ])^{\frac{1}{k}} < \epsilon$,  where $\Pi(\mu, \pi)$ denotes the set of all possible couplings of the distributions $\mu$ and $\pi$, does not imply a bound on the infinity-distance $\sup_{\theta \in S} \left|\log \frac{\nu(\theta)}{\pi(\theta)}\right|$. 
For instance, consider $\pi \propto \mathrm{unif}[0,1]$, and $\nu$ the uniform distribution on the grid $\{1,2,\ldots, n\}$ for $n< \frac{1}{4\epsilon}$.
And consider the coupling $\Pi(\pi,\nu)$ of $\pi$ and $\nu$, which, for every $i \in\{1,2,\ldots, n\}$, transports all the probability  mass of $\pi$ in the interval $(\frac{i}{n}, \frac{i+1}{n}]$ to a point mass at $\frac{i+1}{n}$.
This coupling does not transport any of the probability mass a distance of more that $\frac{1}{2} \epsilon$, and hence $W_k(\pi, \nu) < \epsilon$ for any $k \in \mathbb{N} \cup \{\infty\}$.
On the other hand, since $\nu$ has atomic point-masses while $\pi$ is a continuous distribution, we have $\mathrm{d}_{\infty}(\pi, \mu)= \sup_{\theta \in S} \left|\log \frac{\nu(\theta)}{\pi(\theta)}\right| = \infty$.

In some cases, a contraction can be shown to occur with probability $1$, yielding a sample with bounds in the $\infty$-Wasserstein metric.
For instance, this is the case for ``idealized'' versions of the Hamiltonian Monte Carlo Markov chain whose steps are determined by continuous trajectories determined by the Hamiltonian mechanics \cite{mangoubi2017rapidp1,vishnoi_HMC}.
For this idealized version of the Hamiltonian Monte Carlo Markov chain, in the special case where the target log-density is strongly convex and smooth on all of $\mathbb{R}^d$, one can oftentimes show that the two chains contract to within a Euclidean distance of $O(\epsilon)$, and generate a sample $\hat{\theta}$ from $\pi$ with an error of $O(\epsilon)$ in the $\infty$-Wasserstein metric, after a number of Markov chain steps that is logarithmic in $\frac{1}{\epsilon}$.
From such $\hat{\theta}$, if $f$ is $1$-Lipschitz on $\mathbb{R}^d$, one can obtain a sample from $\pi$ with $O(\epsilon d)$ infinity-distance error by adding a uniform random vector on a ball of radius roughly $\epsilon d.$
One can extend this approach to the problem of sampling from, e.g., smooth and Lipschitz convex log-densities supported on a polytope, by extending $f$ to a $L$-Lipschitz log-density on all of $\mathbb{R}^d$, and adding a strongly convex regularizer.
However, to implement such a Markov chain as an algorithm, the trajectories must be discretized, and the number of discretization steps to bring the two Markov chains within $O(\epsilon)$ Euclidean distance is polynomial in $\frac{1}{\epsilon}$ (a polynomial dependence on $\epsilon$ also occurs when one only seeks an $O(\epsilon)$ bound on the $1$- or $2$- Wasserstein error for many Markov chain algorithms via contractive couplings, including ``Unadjusted'' Langevin dynamics Markov chains \cite{durmus2017nonasymptotic, dalalyan2017theoretical}).
Thus, we need a different approach if we wish to achieve $O(\epsilon)$ infinity-divergence bounds in runtime logarithmic in $\frac{1}{\epsilon}$.

Another approach would be to design a coupling such that the two Markov chains contract within some distance $O(\frac{1}{d})$, and then to propose to add a uniform random vector on a ball of radius roughly $\Theta(1)$ to each Markov chain and accept this proposed step according to the Metropolis acceptance rule for $\pi$.
If $f$ is $1$-Lipschitz on $\mathbb{R}^d$ (or on the constraint set $K$), the acceptance probability will be at least $\frac{1}{2}$ and one can show that the total variation distance of two Markov chains is  at least $\frac{1}{2}$ after some number $T$ steps where $T$ is polynomial in $\frac{1}{\epsilon}$.
Repeating this coupling every $T$ steps, one can show that the total variation distance of the two chains decreases by a factor of at least $\frac{1}{2}$ every $T$ steps, allowing one to generate a sample from $\pi$ with total variation error $O(\epsilon)$ in a number of steps that is logarithmic in $\frac{1}{\epsilon}$.
In other words, there exists a coupling of the two chains such that after a number of steps that is logarithmic in $\frac{1}{\epsilon}$, the algorithm's Markov chain is equal to the chain with distribution $\pi$ with probability $1-O(\epsilon)$.
However, we may still have that, with probability roughly $\epsilon$, the algorithm's Markov chain is concentrated in a region of space of volume $O(\epsilon)$ where the total probability mass of $\pi$ is much smaller than $\epsilon$, which would mean that the infinity-distance to $\pi$ would be $\Omega(1)$ even though the total variation distance to $\pi$ is $O(\epsilon)$.

\subsection{Challenges in extending  methods based on gradient flows in space of distributions from KL bounds to infinity-distance bounds}\label{sec_challenges_KL}
In addition to Markov chain bounds achieved via probabilistic coupling methods, bounds for certain Markov chains can also be achieved by analyzing, e.g., the Langevin diffusion process as a gradient flow of the KL divergence functional (also called the relative entropy functional) in the space of probability distributions under the $2$-Wasserstein metric.
Using this approach, one can show that the Langevin diffusion process with stationary distribution $\pi$ converges to within KL divergence distance $\epsilon$ of $\pi$ in (continuous) time $t$ that is logarithmic in $\frac{1}{\epsilon}$, if, for instance, $\pi$ is strongly log-concave and smooth (see e.g. \cite{villani2009optimal, jordan1998variational}).
One can then discretize the Langevin diffusion process using a discrete-time Markov chain algorithm (such as the Langevin dynamics Markov chain), and bound the distance between the distribution of the algorithm's Markov chain and the diffusion process in the KL or Renyi divergence metrics (see for instance \cite{wibisono2018sampling, cheng2018convergence, durmus2019analysis} for Langevin dynamics Markov chains).
However, these bounds are polynomial in $\frac{1}{\epsilon}$ rather than logarithmic in $\frac{1}{\epsilon}$. 
The polynomial dependence on $\epsilon$ is due to the fact that the discretization step size for the Langevin dynamics algorithms to approximate the Langevin diffusion process  with error $O(\epsilon)$ is polynomial in $\frac{1}{\epsilon}$.

 For instance, if one wishes to sample from $\pi \sim e^{-f(\theta)}$ with support on the interval $[-2,2]$, one can use the Langevin dynamics Markov chain, which is a (first-order) ``Euler'' discretization of the Langevin diffusion on $\mathbb{R}$ with updates $\hat{\theta}_{i+1}$ at each step $i+1$ given as follows:  $\hat{\theta}_{i+1} = \hat{\theta}_i - \eta \nabla f(\hat{\theta}_i) + \sqrt{2\eta} \xi$, where $\xi \sim N(0, I_d)$.
To (approximately) sample from $\pi \sim e^{-f(\theta)}$ with support on $[-2,2]$, one could then output only those steps of the Markov chain which fall inside the constraint interval $[-2,2]$.
 If $f:\mathbb{R}^1 \rightarrow \mathbb{R}$,  $f(\theta) = \frac{1}{2}\theta^2$, then the Langevin diffusion is $\mathrm{d}\theta = -\theta \mathrm{d}t +\sqrt{2}\mathrm{d}W_t$, and the discretization is  $\hat{\theta}_{i+1} = \hat{\theta}_i - \eta \hat{\theta}_i + 2\sqrt{\eta} \xi$. 
 The solution to the (continuous-time) Langevin diffusion at time $t$ is a Gaussian random variable $\theta_t \sim \mathcal{N}(\theta_0 e^{-t},  1-e^{-2t})$, and its stationary distribution is the target distribution $\propto e^{-\frac{1}{2}\theta^2}$. 
On the other hand, the stationary distribution of the discrete-time Markov chain with step size parameter $\eta$ is $N(0,\frac{1}{1-\frac{1}{2} \eta})$.
Therefore, to have the Langevin Markov chain approximate the Langevin diffusion to within infinity-distance error $O(\epsilon)$ (even just in a compact constraint interval such as $K= [-2,2]$), we would need to have a step size $\eta = \mathrm{poly}(\epsilon)$, and the number of Markov chain steps required to sample within infinity-distance error $O(\epsilon)$ would be {\em polynomial} in $\frac{1}{\epsilon}$.

\section{Proofs of Theorem  \ref{thm_TV_to_inf_divergence} and Theorem  \ref{thm_infinity_divergence_sampler}}\label{sec_TV_proof}

In this section, we first prove our main technical result, Theorem \ref{thm_TV_to_inf_divergence} (Appendix \ref{sec_completing_proof_of_TV_to_inf}).
The Lemmas we use to prove Theorem  \ref{thm_TV_to_inf_divergence} are proved in Appendix \ref{sec_stretching}, \ref{sec_acceptance_probability}, and \ref{sec_intinity_distance_bounding}.
 Finally, we plug in TV bounds for the Dikin Walk Markov chain  \cite{narayanan2017efficient} to Theorem  \ref{thm_TV_to_inf_divergence} to complete the proof of Theorem \ref{thm_infinity_divergence_sampler} (Appendix \ref{sec_proof_thm_infinity_divergence_sampler}).

In the following, we define the random variable $\tau$ to be the number of iterations of the ``for'' loop started by Algorithm \ref{alg_TV_to_pure} if Algorithm \ref{alg_TV_to_pure} halts while it is running the for loop.
Otherwise, we set $\tau = \tau_{\mathrm{max}}+1$.

\paragraph{Setting the parameters.}
In the following, we assume that 
\begin{enumerate}
\item $\tau_{\mathrm{max}} \geq 5d \log(\frac{R}{r}) + 5 LR + \eps$,
\item  $\Delta \leq \frac{\epsilon}{512\tau_{\mathrm{max}} \max(d, L R)}$, \item  and  $\delta \leq \frac{1}{64}\epsilon\times (\frac{R}{\Delta r})^{-d} e^{-LR}$.
\end{enumerate}

\subsection{Bounding the number of iterations and completing the proof of Theorem \ref{thm_TV_to_inf_divergence} }\label{sec_completing_proof_of_TV_to_inf}

\begin{proof}\textbf{[of Theorem \ref{thm_TV_to_inf_divergence}]}

\paragraph{Correctness.}
By Lemmas \ref{lemma_lower_bound} and \ref{lemma_upper_bound}, we have that 
\begin{enumerate}
    \item 
    The distribution $\nu$ of the output $\hat{\theta}$ of Algorithm \ref{alg_TV_to_pure} satisfies  $\mathrm{d}_{\infty}(\nu, \pi) \leq \epsilon$.
    \item Moreover, the distribution $\hat{\nu}$ of $\hat{\theta}$ conditional on $\tau \leq t$ satisfies $\mathrm{d}_{\infty}(\hat{\nu}, \pi) \leq \epsilon$ for any $t< \tau_{\mathrm{max}}$.

\end{enumerate}

\paragraph{Bounding the number of operations.}
Each iteration of Algorithm \ref{alg_TV_to_pure} requires one call to the sampling oracle (Line \ref{Line_sampling_oracle}) for $\mu$, and one call to the membership oracle for $K$ (Line \ref{Line_membership_oracle}). 
In addition, each line of  Algorithm \ref{alg_TV_to_pure} requires no more than $O(d)$ arithmetic operations.
Thus, each iteration of Algorithm \ref{alg_TV_to_pure} can be computed in one call to the sampling oracle for $\mu$, one call to the membership oracle for $K$, plus $O(d)$ arithmetic operations.
The number of iterations $\tau$ is random, and can be bounded as follows: 

By Corollary \ref{cor_runtime} we have that $\mathbb{E}[\tau] \leq 3$ and $\mathbb{P}(\tau \geq t) \leq \left(\frac{2}{3}\right)^t$.
By Corollary \ref{cor_runtime} we also have that
\begin{equation*}\left(\frac{1}{2}\right)^t \leq \mathbb{P}(\tau \geq t) \leq \left(\frac{1}{2} + \frac{\epsilon}{8\tau_{\mathrm{max}}}\right)^t \qquad \forall t\leq \tau_{\mathrm{max}},
\end{equation*}
and, hence, that

\begin{eqnarray*}
\left(\frac{1}{2}\right)^t &\leq & \mathbb{P}(\tau \geq t) \leq \left(\frac{1}{2}\right)^t \times \left(1+\frac{\epsilon}{4\tau_{\mathrm{max}}}\right)^t  \quad \quad \forall t\leq \tau_{\mathrm{max}}\\
&\leq &\left(\frac{1}{2}\right)^t \times e^{\frac{\epsilon}{4}} \quad \quad \forall t\leq \tau_{\mathrm{max}}.
\end{eqnarray*}
Thus, we have that, for all $t\leq \tau_{\mathrm{max}}$,
\begin{equation} \label{eq_a7}
\mathbb{P}(\tau = t) = \mathbb{P}(\tau \geq t) - \mathbb{P}(\tau \geq t+1) \leq  \left(\frac{1}{2}\right)^t e^{\frac{\epsilon}{4}} - \left(\frac{1}{2}\right)^{t+1}
\end{equation}
and that
\begin{equation} \label{eq_a8}
\mathbb{P}(\tau = t) = \mathbb{P}(\tau \geq t) - \mathbb{P}(\tau \geq t+1) \geq  \left(\frac{1}{2}\right)^t - \left(\frac{1}{2}\right)^{t+1} e^{\frac{\epsilon}{4}}.
\end{equation}
Thus, we have that
\begin{equation*} 
\left(\frac{1}{2}\right)^t \left(1- \frac{1}{2}  e^{\frac{\epsilon}{4}}\right) \stackrel{\textrm{Eq. }\eqref{eq_a8}}{\leq}  \mathbb{P}(\tau = t)  \stackrel{\textrm{Eq. }\eqref{eq_a7}}{\leq} \left(\frac{1}{2}\right)^t \left( e^{\frac{\epsilon}{4}}-\frac{1}{2} \right)\qquad \forall t\leq \tau_{\mathrm{max}},\\
\end{equation*}
and, hence, that
\begin{equation} \label{eq_a6}
\left(\frac{1}{2}\right)^t e^{-\frac{\eps}{2}} \leq  \mathbb{P}(\tau = t)  \leq \left(\frac{1}{2}\right)^t e^{\frac{\eps}{2}}\qquad \forall t\leq \tau_{\mathrm{max}},\\
\end{equation}
since $\epsilon \leq 1$.

\end{proof}

\subsection{Stretching the polytope to avoid samples near the boundary} \label{sec_stretching}

\begin{definition}[Interior]\label{def:interior}
For any $\Delta \geq 0$ and any $S \subseteq \mathbb{R}^d$, we define the $\Delta$-interior  of $S$, $\mathrm{int}_{\Delta}(S)$, as
\begin{equation*}
\mathrm{int}_{\Delta}(S) = \{z \in S : B(z,\Delta) \in S\}.
\end{equation*}
\end{definition}

\begin{lemma} \label{Lemma_cvx_hull}
Let $Z \in \mathbb{R}^d$ and $0\leq \Delta \leq \frac{1}{2}$.  Then if $\frac{1}{1-\Delta} Z \in K$, we also have that $Z \in \mathrm{int}_{\Delta r}(K)$.
\end{lemma}

\begin{proof}
Let $\mathcal{C}$ be the convex hull of $B(0,r) \cup \left\{\frac{1}{1-\Delta} Z \right\}$.
Since $B(0,r) \subseteq  K$ and $\frac{1}{1-\Delta} Z \in 
K$, we have that the convex hull $\mathcal{C} \subseteq K$.

Let $h := \max\left\{ s>0 : B\left(Z, \, \, s\right) \in \mathcal{C} \right\}$.  Defining the point $\hat{\theta} := \frac{1}{1-\Delta} Z$, the point $p$ to be a point such that the line $p\hat{\theta}$ is tangent to $B(0,r)$, and the point $q$ to be a point on $p\hat{\theta}$ which minimizes the distance $\|q-Z\|_2$ (see  Figure \ref{convex_hull_diagram}).
Then we have that the triangles $0p\hat{\theta}$ and $Z q \hat{\theta}$ are similar triangles, since $\angle 0 p \hat{\theta}$ and $\angle Z q \hat{\theta}$ are both right angles.  Thus, $\frac{\|Z - q\|_2}{r} = \frac{\|Z - \hat{\theta}\|_2}{\|\hat{\theta}-0\|_2}$.

Therefore,
\begin{eqnarray*}
h & =& \|Z - q\|_2\\
& =& \frac{\|Z - \hat{\theta}\|_2}{\|\hat{\theta}-0\|_2} \times r\\
& =& \frac{\|Z - \frac{1}{1-\Delta} Z\|_2}{\|\frac{1}{1-\Delta} Z\|_2} \times r\\
& =& \frac{\|\frac{\Delta}{1-\Delta} Z\|_2}{\|\frac{1}{1-\Delta} Z\|_2} \times r\\
& = & \Delta \times r\\
\end{eqnarray*}
Thus, $Z \in \mathrm{int}_{\Delta r}(\mathcal{C}) \subseteq \mathrm{int}_{\Delta r}(K).$ 

\end{proof}

\subsection{Bounding the acceptance probability}\label{sec_acceptance_probability}
The following lemma allows us to bound the expected number of oracle calls in Algorithm \ref{alg_TV_to_pure}.

\begin{lemma} \label{lemma_rejection_probability}
For any $0 \leq \Delta \leq \frac{1}{4}$, we have that

\begin{equation*}
    \mathbb{P}_{Z \sim \pi}((Z \in (1-\Delta) K) \geq (1- \Delta)^d e^{-2L \Delta R}.
\end{equation*}
Hence, since by Lemma \ref{Lemma_cvx_hull} $(1-\Delta) K \subseteq \mathrm{int}_{\Delta r}(K)$, we also have that
\begin{equation*}
    \mathbb{P}_{Z \sim \pi}(Z \in \mathrm{int}_{\Delta r}(K)) \geq (1- \Delta)^d e^{-2L \Delta R}.
\end{equation*}

\end{lemma}

\begin{proof}
Let $c = (\int_{K} e^{-f\left(\theta\right)} \mathrm{d}\theta)^{-1}$ be the normalizing constant of $\pi$, that is, $\pi(\theta) = c e^{-f(\theta)}$ for $\theta \in K$.
And let $\tilde{\pi}$ be the distribution
\begin{equation*}
    \tilde{\pi}(\theta) = \begin{cases} 
     \tilde{c} e^{-f\left(\frac{1}{1- \Delta}\theta\right)} & \textrm{ if } \theta \in (1-\Delta)K \\
      0 & \textrm{otherwise} 
   \end{cases}
    \end{equation*}
where $\tilde{c} = (\int_{(1-\Delta)K} e^{-f\left(\frac{1}{1- \Delta}\theta\right)} \mathrm{d}\theta)^{-1}$ is the normalizing constant of $\tilde{\pi}$.
In other words, if $Z \sim \pi$ then we have $(1-\Delta)Z \sim \tilde{\pi}$.
Then 

\begin{equation} \label{eq_a1}
    \tilde{c} = \left(\frac{1}{1- \Delta}\right)^d c.
\end{equation}
Thus,
\begin{equation} \label{eq_a2}
    \frac{\tilde{\pi}(\theta)}{\pi(\theta)} = \frac{\tilde{c}}{c} e^{f(\theta) - f\left(\frac{1}{1- \Delta}\theta\right)} = \left(\frac{1}{1- \Delta}\right)^d e^{f(\theta) - f\left(\frac{1}{1- \Delta}\theta\right)} \qquad \forall \theta \in (1-  \Delta) K.
    \end{equation}
Since $K \subseteq B(0,R)$, we have that 
\begin{equation} \label{eq_a2b}
    \left\|\frac{1}{1-\Delta}\theta - \theta \right\|_2 = \frac{\Delta}{1-\Delta} \|\theta \|_2 \leq 2\Delta \|\theta \|_2  \leq 2\Delta R,
\end{equation}
for all $\theta \in K$.
Therefore, since $f$ is $L$-Lipschitz, Equations \eqref{eq_a1} and \eqref{eq_a2} imply that
\begin{equation} \label{eq_a3}
    \left(\frac{1}{1- \Delta}\right)^d e^{- 2L \Delta R} \leq \frac{\tilde{\pi}(\theta)}{\pi(\theta)}  \leq  \left(\frac{1}{1- \Delta}\right)^d e^{2L \Delta R} \qquad \forall \theta \in (1-  \Delta) K.
\end{equation}
Therefore,

\begin{align*}
    \mathbb{P}_{Z \sim \pi}(Z \in (1-\Delta)K)
    &= \int_{\theta \in (1-\Delta) K} \pi(\theta) \mathrm{d} \theta\\
    &\stackrel{\textrm{Eq. }\ref{eq_a3}}{\geq}  \int_{\theta \in (1-\Delta) K} (1- \Delta)^d e^{- 2L \Delta R} \tilde{\pi}(\theta) \mathrm{d} \theta\\
    &=  (1-\Delta)^d e^{-2L \Delta R} \int_{\theta \in (1-\Delta) K} \tilde{\pi}(\theta) \mathrm{d} \theta\\
    &= (1- \Delta)^d e^{- 2L \Delta R}.
\end{align*}
\end{proof}

\begin{corollary} \label{corr_rejection_probability}
 For any $0 \leq \Delta \leq \frac{1}{2}$, we have that

\begin{equation*}
    \mathbb{P}_{Z \sim \pi}(Z \in \mathrm{int}_{\Delta r}((1-\Delta)K)) \geq [(1- \Delta)^d e^{- 2L \Delta R}]^2.
\end{equation*}
\end{corollary}

\begin{proof}
First, we apply Lemma \ref{lemma_rejection_probability} to the convex body $K$ and the distribution $\pi$ to show that
\begin{equation}\label{eq_c1}
 \mathbb{P}_{Z \sim \pi}((Z \in (1-\Delta)K) \geq (1- \Delta)^d e^{- 2L \Delta R}.
 \end{equation}
 Next, we let $\pi^{\dagger}(\theta) \propto e^{-f(\theta)} \mathbbm{1}\{\theta \in (1-\Delta) K\}$, and we apply  Lemma \ref{lemma_rejection_probability} again, but this time to the convex body $(1-\Delta)K$ (which, like $K$, is contained in $B(0,R)$) and the distribution $\pi^{\dagger}$ (which, like $\pi$, has $L$-Lipschitz log-density) to show that
 \begin{equation} \label{eq_c2}
    \mathbb{P}_{W \sim \pi^{\dagger}}(W \in \mathrm{int}_{\Delta r}((1-\Delta) K)) \geq (1- \Delta)^d e^{- 2L \Delta R}.
\end{equation}
Thus,

\begin{align*}
 &\mathbb{P}_{Z \sim \pi}(Z \in \mathrm{int}_{\Delta r}((1-\Delta) K)) 
 \\&=  \mathbb{P}_{Z \sim \pi}\left(Z \in \mathrm{int}_{\Delta r}((1-\Delta) K) \bigg |  Z \in (1-\Delta)K \right) \times  \mathbb{P}_{Z \sim \pi}((Z \in (1-\Delta)K)
 \\
 & =  \mathbb{P}_{W \sim \pi^{\dagger}}(W \in \mathrm{int}_{\Delta r}((1-\Delta)K)) \times \mathbb{P}_{Z \sim \hat{\pi}}((Z \in (1-\Delta)K)\\
 & \stackrel{\textrm{Eq. }\eqref{eq_c1}, \eqref{eq_c2}}{\geq}  [(1- \Delta)^d e^{- 2L \Delta R}]^2.
 \end{align*}
 \end{proof}

\begin{corollary} \label{cor_runtime}
Algorithm \ref{alg_TV_to_pure} finishes in $\tau$ calls to the sampling oracle, where 
\begin{equation*}
    \left(\frac{1}{2}\right)^t \leq \mathbb{P}(\tau \geq t) \leq \left(\frac{1}{2} + \frac{\epsilon}{8\tau_{\mathrm{max}}}\right)^t \leq \left(\frac{2}{3}\right)^t \qquad \qquad \forall t \in [\tau_{\mathrm{max}}],
\end{equation*}
and, hence,
\begin{equation*}
    \mathbb{E}[\tau] \leq 3.
\end{equation*}

\end{corollary}

\begin{proof}
By Corollary \ref{corr_rejection_probability}, we have that, conditional on Algorithm \ref{alg_TV_to_pure} reaching some iteration $i \in \mathbb{N}$, the probability that Algorithm $i$ will reject $\hat{\theta}$ at step $i$ is
\begin{align} \label{eq_a4}
      \mathbb{P}(\hat{\theta} \textrm{ rejected at step $i$}| \textrm{step $i$ reached}) &\leq \frac{1}{2} + \frac{1}{2} \left(1- \mathbb{P}_{Z \sim \pi}(Z \in \mathrm{int}_{\Delta r}((1-\Delta)K))\right)\nonumber\\
      &\stackrel{\textrm{Corr. }\ref{corr_rejection_probability}}{\leq}  \frac{1}{2} + \frac{1}{2} \left(1-  [(1- \Delta)^d e^{- 2L \Delta R}]^2 \right) \nonumber\\
      &\leq \left(\frac{1}{2} + \frac{\epsilon}{8\tau_{\mathrm{max}}}\right)^t\nonumber\\
      &\leq \frac{2}{3},
\end{align}
where the second-to-last inequality holds because $\Delta \leq \frac{\epsilon}{512\tau_{\mathrm{max}} \max(d, L R)}$.
The last inequality holds because $\tau_{\mathrm{max}} \geq \frac{\epsilon}{3}$.
Hence,
\begin{eqnarray} \label{eq_a5}
    \mathbb{P}(\tau \geq t) &\leq  \Pi_{i=1}^t \mathbb{P}(\hat{\theta} \textrm{ rejected at step $i$}| \textrm{step $i$ reached})\nonumber\\
    & \stackrel{\textrm{Eq. }\eqref{eq_a4}}{\leq}   \left(\frac{2}{3}\right)^t  \qquad \qquad \forall t \geq 0,
\end{eqnarray}
and
\begin{align}
    \mathbb{E}[\tau] &\leq \sum_{t=1}^{\infty} \mathbb{P}(\tau \leq t)\nonumber\\
    & = \sum_{t=1}^{\infty} \mathbb{P}(\tau \geq t)\nonumber\\
    & \stackrel{\textrm{Eq. }\ref{eq_a5}}{\leq} \sum_{t=1}^{\infty} \left(\frac{2}{3}\right)^t\nonumber\\
    &= \frac{1}{1-\frac{2}{3}}\nonumber\\
    &= 3.
\end{align}
\end{proof}

\noindent

\subsection{Bounding the infinity-distance error}\label{sec_intinity_distance_bounding}

The next lemma  allows us to provide a lower bound on the density $\nu$ of the output $\hat{\theta}$.
Let $\tilde{\nu}$ be the distribution of the random variable  $Z = y  + \Delta r \xi$, where $y \sim \mu$, and $\xi \sim \mathrm{unif}(B(0,1))$.
Let $\nu^\ast$ be the distribution of $\left(\frac{1}{1- \Delta}\right)Z$ conditional on the event that $\left(\frac{1}{1- \Delta}\right)Z \in K$.

\begin{lemma} \label{lemma_lower_bound}
For every $\theta \in K$ we have
\begin{equation*}
    \nu^\ast(\theta)  \geq  e^{-\frac{\epsilon}{2}}\pi(\theta), 
\end{equation*}
and
\begin{equation*}
\nu(\theta) \geq e^{-\epsilon} \pi(\theta).
\end{equation*}

\end{lemma}

\begin{proof}
For all $\theta \in \mathrm{int}_{\Delta r}(K)$, we have
\begin{align} \label{eq_b3}
    \tilde{\nu}(\theta) &= \frac{1}{\mathrm{Vol}(B(0, \Delta r))} \int_{w \in B(0, \Delta r)} \mu(\theta + w) \mathrm{d}w \nonumber \\
    &\geq \frac{1}{\mathrm{Vol}(B(0, \Delta r))} \left[ \int_{w \in B(0, \Delta r)}  \pi(\theta + w) \mathrm{d}w  - \delta \right] \nonumber\\
    & \geq \pi(\theta) e^{-L \Delta r} - \frac{\delta}{\mathrm{Vol}(B(0, \Delta r))} \nonumber\\
       & = \pi(\theta) e^{-L \Delta r} - \frac{\delta}{\mathrm{Vol}(B(0, \Delta r))} \times \pi(\theta) \times \frac{1}{\pi(\theta)}\nonumber\\
    & \geq  \pi(\theta) e^{-L \Delta r} - \frac{\delta}{\mathrm{Vol}(B(0, \Delta r))} \times \pi(\theta) \times  \left(  \frac{\max_{w \in K} \pi(w)}{\min_{w \in K} \pi(w)} \times \mathrm{Vol}(B(0, R)) \right )\nonumber\\
    & \geq \pi(\theta) e^{-L \Delta r} - \left(\frac{R}{\Delta r}\right)^d \times \delta \times \pi(\theta) \times e^{L R}\nonumber\\
    & =  \pi(\theta) \times \left [ e^{-L \Delta r} - \left(\frac{R}{\Delta r}\right)^d \times \delta \times e^{L R} \right ]\nonumber\\
    & \geq \pi(\theta) \times e^{-\frac{\epsilon}{8}}.
\end{align}
Where the last inequality holds since $\Delta \leq \frac{\epsilon}{16Lr}$ and $\delta \leq (e^{-\frac{\epsilon}{16}} - e^{-\frac{\epsilon}{8}}) \times \left(\frac{R}{\Delta r}\right)^{-d} e^{-L R}$.
Thus, with probability at least $1-\hat{\delta}$, we have that the conditional distribution  $\tilde{\nu}$ satisfies:
\begin{equation} \label{eq_b1}
\tilde{\nu}(\theta) \geq \pi(\theta) \times e^{-\frac{\epsilon}{8}} \qquad \forall    \theta \in \mathrm{int}_{\Delta r}(K).
\end{equation}
Recall that $\tilde{\nu}$ is the distribution of the random variable  $Z = y  + \Delta r \xi$, where $y \sim \mu$.
And recall that $\nu^\ast$ is the distribution of $(\frac{1}{1- \Delta})Z$ conditional on the event that $(\frac{1}{1- \Delta})Z \in K$.
Then Inequality \eqref{eq_b1} implies that $\forall    \theta \in K$

\begin{align} \label{eq_b2}
\nu^\ast(\theta) &\geq (1- \Delta)^d \tilde{\nu}((1- \Delta)\theta) \nonumber\\
&\stackrel{\textrm{Eq. }\eqref{eq_b1} \textrm{, Lemma } \ref{Lemma_cvx_hull}}{\geq} (1- \Delta)^d \pi((1- \Delta)\theta) \times e^{\frac{-\epsilon}{8}}\nonumber\\
&\geq \pi((1- \Delta)\theta) e^{-\frac{\epsilon}{4}} \nonumber \\
&\geq \pi(\theta) e^{-\frac{\epsilon}{2}}, 
\end{align}
where the second Inequality holds by Inequality \eqref{eq_b1} since $(1- \Delta)\theta \in  \mathrm{int}_{\Delta r}(K)$  by Lemma \ref{Lemma_cvx_hull}. 
 The third inequality holds because $\Delta \leq \frac{\min(1, \epsilon)}{16d}$.
The last inequality holds since $\Delta \leq \frac{\epsilon}{16L R}$ since $K \subseteq B(0,R)$ implies that $\|\Delta \theta\| \leq \Delta R$.

Finally, Corollary \ref{cor_runtime} implies that
\begin{eqnarray*}
    \nu(\theta) &\geq& \mathbb{P}(\tau < \tau_{\mathrm{max}}) \times \nu^\ast(\theta)\\
    &\stackrel{\textrm{Corr. }\ref{cor_runtime}}{\geq}& \left(1- \left(\frac{2}{3}\right)^{\tau_{\mathrm{max}}}\right) \times \nu^\ast(\theta)\\
    &\geq& e^{-\frac{\epsilon}{2}} \times \nu^\ast(\theta)\\
    &\stackrel{\textrm{Eq. }\ref{eq_b2}}{\geq}&  e^{-\epsilon} \times \nu(\theta),
\end{eqnarray*}
where the third inequality holds since $\tau_{\mathrm{max}} \geq \log( \frac{4}{\eps})$ and $\eps \leq 1$.

\end{proof}

\begin{lemma} \label{lemma_upper_bound}
For every $\theta \in K$ we have
\begin{equation*}
    \nu^\ast(\theta) \leq e^{\frac{\epsilon}{2}} \pi(\theta), 
\end{equation*}
and
\begin{equation*}
\nu(\theta) \leq e^{\epsilon} \pi(\theta).
\end{equation*}

\end{lemma}

\begin{proof}
For all $\theta \in K$, we have

\begin{align} \label{eq_d1}
   \tilde{\nu}(\theta) &= \frac{1}{\mathrm{Vol}(B(0, \Delta r))} \int_{w \in B(0, \Delta r)} \mu(\theta + w) \mathrm{d}w \nonumber\\
    &\leq \frac{1}{\mathrm{Vol}(B(0, \Delta r))} \left[ \int_{w \in B(0, \Delta r)}  \pi(\theta + w) \mathrm{d}w  + \delta \right] \nonumber\\
    & \leq \pi(\theta) e^{L \Delta r} + \frac{\delta}{\mathrm{Vol}(B(0, \Delta r))}\nonumber\\
       & = \pi(\theta) e^{L \Delta r} + \frac{\delta}{\mathrm{Vol}(B(0, \Delta r))} \times \pi(\theta) \times \frac{1}{\pi(\theta)} \nonumber\\
    & \leq  \pi(\theta) e^{L \Delta r} + \frac{\delta}{\mathrm{Vol}(B(0, \Delta r))} \times \pi(\theta) \times  \left(  \frac{\max_{w \in K} \pi(w)}{\min_{w \in K} \pi(w)} \times \mathrm{Vol}(B(0, R)) \right )\nonumber\\
    & \leq \pi(\theta) e^{L \Delta r} + \left(\frac{R}{\Delta r}\right)^d \times \delta \times \pi(\theta) \times e^{L R}\nonumber\\
    & =  \pi(\theta) \times \left [ e^{L \Delta r} + \left(\frac{R}{\Delta r}\right)^d \times \delta \times e^{L R} \right ]\nonumber\\
    & \leq \pi(\theta) \times e^{\frac{\epsilon}{8}},
\end{align}
where the last inequality holds since $\Delta \leq \frac{\epsilon}{16Lr} $ and $\delta \leq (e^{\frac{\epsilon}{16}} - e^{\frac{\epsilon}{8}}) \times (\frac{R}{\Delta r})^{-d} e^{-LR}$.

Moreover, by Corollary \ref{corr_rejection_probability} and Inequality \eqref{eq_b3} we have that
\begin{align} \label{eq_d2}
    \mathbb{P}_{Z \sim \tilde{\nu}}(Z \in \mathrm{int}_{\Delta r}((1-\Delta)K)) 
    &\stackrel{\textrm{Eq. }\eqref{eq_b3}}{\geq} e^{\frac{-\epsilon}{8}}  \times  \mathbb{P}_{Z \sim \pi}(Z \in \mathrm{int}_{\Delta}((1-\Delta)K)) \nonumber\\
    &\stackrel{\textrm{Corr. }\ref{corr_rejection_probability}}{\geq} e^{\frac{-\epsilon}{8}}  \times [(1- \Delta)^d e^{- L \Delta R}]^2 \nonumber\\
    & \geq  e^{\frac{-\epsilon}{4}}, 
\end{align}
where the last inequality holds since  $\Delta \leq \frac{\min(1, \epsilon)}{64d}$ and $\Delta \leq \frac{\epsilon}{128 LR}$.
Thus, Inequality \eqref{eq_d2} implies that
\begin{equation} \label{eq_d3}
    \int_{\hat{K}} \nu^\ast(\theta) \mathrm{d} \theta \leq e^{\frac{\epsilon}{4}}  \int_{\hat{K}} \tilde{\nu}(\theta) \mathrm{d} \theta.
\end{equation}
Recall that $\tilde{\nu}$ is the distribution of the random variable  $Z = y  + \Delta r \xi$, where $y \sim \mu$. 
And recall that $\nu^\ast$ is the distribution of $(\frac{1}{1- \Delta})Z$ conditional on the event that $(\frac{1}{1- \Delta})Z \in K$.
Therefore, inequality \eqref{eq_d3} implies that, for all $\theta \in K$, 

\begin{align} \label{eq_d4}
    \nu^\ast(\theta) &=  \frac{\int_{K} \nu^\ast(z) \mathrm{d} z}{\int_{K} \tilde{\nu}(z) \mathrm{d} z} \tilde{\nu}((1-\Delta)\theta) \nonumber \\
&\stackrel{\textrm{Eq. }\ref{eq_d3}}{\leq}     e^{\frac{\epsilon}{4}} \times \tilde{\nu}((1- \Delta)\theta) \nonumber\\
&\stackrel{\textrm{Eq. }\ref{eq_d1}}{\leq}  e^{\frac{\epsilon}{4}} \times e^{\frac{\epsilon}{8}} \pi((1- \Delta)\theta) \nonumber\\
&\leq e^{\frac{3\epsilon}{8}} \times e^{\Delta L R}  \pi(\theta) \nonumber\\
&\leq e^{\frac{\epsilon}{2}} \pi(\theta),
\end{align}
where the second-to-last inequality holds because $\pi \propto e^{-f}$ where $f$ is $L$-Lipschitz, and since $\|(1- \Delta)\theta - \theta\| =  \Delta\|\theta\| \leq \Delta R$ because $K \subseteq B(0,R)$.
And the last inequality holds because $\Delta \leq \frac{\epsilon}{32 L R}$.

Finally, Corollary \ref{cor_runtime} implies that
\begin{eqnarray*}
    \nu(\theta) &\leq& \nu^\ast(\theta) + \mathbb{P}(\tau = \tau_{\mathrm{max}}) \times \frac{1}{B(0,r)}\\
    &\stackrel{\textrm{Corr. }\ref{cor_runtime}}{\leq}& e^{\frac{\epsilon}{2}} \pi(\theta) + \left(\frac{2}{3}\right)^{\tau_{\mathrm{max}}} \times \frac{1}{B(0,r)}\\
   &\leq & e^{\frac{\epsilon}{2}} \pi(\theta) + \left(\frac{2}{3}\right)^{\tau_{\mathrm{max}}} \times \frac{1}{B(0,r)} \times \pi(\theta) \times  \left(  \frac{\max_{w \in K} \pi(\theta)}{\min_{w \in K} \pi(\theta)} \times \mathrm{Vol}(B(0, R)) \right )\\
   &\leq & e^{\frac{\epsilon}{2}} \pi(\theta) + \left(\frac{2}{3}\right)^{\tau_{\mathrm{max}}} \times  \left(\frac{R}{r}\right)^d \times e^{L R} \times \pi(\theta)\\
&\leq & e^{\epsilon} \pi(\theta),
\end{eqnarray*}
where the last inequality holds because $\tau_{\mathrm{max}} \geq 5d \log(\frac{R}{r}) + 5 LR + \eps$.

\end{proof}

\subsection{Proof of Theorem \ref{thm_infinity_divergence_sampler}} \label{sec_proof_thm_infinity_divergence_sampler}

\begin{proof}[{\bf of Theorem \ref{thm_infinity_divergence_sampler}}]
 We implement Algorithm \ref{alg_TV_to_pure}, using  the Dikin Walk Markov chain in  \cite{narayanan2017efficient} as a subroutine to compute the TV-bounded sampling oracle for the distribution $\mu$.

To apply Algorithm  \ref{alg_TV_to_pure} and Theorem \ref{thm_TV_to_inf_divergence},  we require that $\|\mu - \pi\|_{\mathrm{TV}} \leq \delta$, where $\delta = \frac{1}{64}\epsilon\times (\frac{R}{\Delta r})^{-d} e^{-LR}$, as well as the following hyperparameter values for Algorithm \ref{alg_TV_to_pure}:
\begin{enumerate}
\item
$\Delta = \frac{\epsilon}{512\tau_{\mathrm{max}} \max(d, L R)}$, and 
\item $\tau_{\mathrm{max}} = 5d \log(\frac{R}{r}) + 5 LR + \eps$.  
\end{enumerate}
To sample from such a distribution $\mu$, we implement the Dikin Walk Markov chain given in Section 3 of \cite{narayanan2017efficient}, with logarithmic-barrier for the polytope $K$ and  hyper-parameters specified by their Condition 2, for $\delta = \frac{1}{64}\epsilon\times (\frac{R}{\Delta r})^{-d} e^{-LR}$.
To provide an initial point $\theta_0$  to the Dikin Walk Markov chain, we sample $\theta_0 \sim \mathrm{unif}(B(0,r))$.
Since $f$ is $L$-Lipschitz and $B(0,r) \subseteq K \subseteq B(0,R)$, the distribution $\mu_0$ of the initial point $\theta_0$ satisfies
\begin{align*}
     \sup_{z\in K} \frac{\mu_0(z)}{\pi(z)} &\leq \frac{1}{\mathrm{Vol}(B(0,r))} \times  \left(  \frac{\max_{z \in K} \pi(\theta)}{\min_{z \in K} \pi(\theta)} \times \mathrm{Vol}(B(0, R)) \right )\\
     & \leq \left(\frac{R}{r}\right)^d \times e^{RL}.
\end{align*}
Thus, the distribution $\mu_0$ of the initial point  $\theta_0$ is $w$-warm with respect to the distribution $\pi$, for $w = \left(\frac{R}{r}\right)^d \times e^{RL}$.

By Lemma 4 of \cite{narayanan2017efficient}, their Dikin Walk Markov chain, with initial point $\theta_0$, outputs a point $\|\mu - \pi\|_{\mathrm{TV}} \leq \delta$ and takes at most 
$$O\left((m^2d^3 + m^2 d L^2 R^2)  \log(\frac{w}{\delta})\right)= O\left((m^2d^3 + m^2 d L^2 R^2) \times \left[LR + d\log\left(\frac{Rd +LRd}{r \eps}\right)\right]\right)$$  steps. 
Moreover,  each iteration takes $O(1)$ function evaluations and $md^{\omega-1}$ arithmetic operations.

\paragraph{Bounding the infinity-distance.}
Since the distribution $\mu$ of the point $\hat{\theta}$ provided by the Dikin Walk Markov chain satisfies $\|\mu - \pi\|_{\mathrm{TV}} \leq \delta$, we have, by Theorem \ref{thm_TV_to_inf_divergence}, that Algorithm \ref{alg_TV_to_pure} outputs a point $\hat{\theta} \in K$, such that the distribution $\nu$ of $\hat{\theta}$ satisfies $\mathrm{d}_{\infty}(\nu, \pi) \leq \epsilon$.

\paragraph{Bounding the number of operations.}
Moreover, by the proof of Theorem \ref{thm_TV_to_inf_divergence}, we also have that  Algorithm \ref{alg_TV_to_pure} finishes in $\tau$ calls to the sampling oracle (computed via the Dikin Walk Markov chain) and $\tau$ calls to the membership oracle, plus $O(\tau d)$ arithmetic operations, where $\mathbb{E}[\tau] \leq 3$ and $\mathbb{P}(\tau \geq t) \leq \left(\frac{2}{3}\right)^t$ for all $\tau \leq \tau_{\mathrm{max}}$.

The membership oracle for $K$ can be computed in $md$ steps since this can be done by checking the inequality $A\theta \leq b$.

Thus,  the total number of steps, when implementing Algorithm \ref{alg_TV_to_pure} with the Dikin Walk Markov chain as subroutine is 
\begin{equation*}
 O\left(\tau \times (m^2d^3 + m^2 d L^2 R^2)  \log(\frac{w}{\delta})\right)
= O\left(\tau \times(m^2d^3 + m^2 d L^2 R^2) \times \left[LR + d\log\left(\frac{Rd +LRd}{r \eps}\right)\right]\right),
\end{equation*}
where $\mathbb{E}[\tau] \leq 3$ and $\mathbb{P}(\tau \geq t) \leq \left(\frac{2}{3}\right)^t$ for all $t \geq 0$, and $\tau \leq \tau_{\mathrm{max}}$ w.p. 1, 
and where each step takes $O(1)$ function evaluations and $md^{\omega-1}$ arithmetic operations.

\end{proof}

\section{Proofs of applications  to differentially private optimization}\label{sec:DP}

\subsection{Proof of Corollary \ref{corr_DP}} \label{Sec_DP_ERM}

To prove Corollary \ref{corr_DP} we will need the following Lemma about the exponential mechanism from \cite{bassily2014private}:

\begin{lemma}[Theorems III.1 and III.2 in \cite{bassily2014private}]\label{Lemma_exponential_mechanism}

Suppose that $\hat{\theta}$ is sampled from the distribution   $\pi(\theta) \propto e^{- \frac{\eps}{2LR}f(\theta, x)}$ where  $f(\theta, x) = \sum_{i=1}^n \ell_i(\theta,x_i)$  where each $\ell_i : K \times \mathcal{D} \rightarrow \mathbb{R}$ is $L$-Lipschitz and convex, and $K \subseteq B(0,R)$ is convex.
Then $\hat{\theta}$ is $\eps$-differentially private and achieves ERM utility
\begin{equation*}
    \mathbb{E}[f(\hat{\theta}, x) - \min_{\theta \in K} f(\theta, x)] = O\left(\frac{d L R}{\eps}\right).
\end{equation*}
\end{lemma}

\begin{proof}[{\bf of Corollary \ref{corr_DP}}]
To prove Corollary \ref{corr_DP}, we first use Algorithm \ref{alg_TV_to_pure} and the Dikin Walk Markov chain from \cite{narayanan2017efficient} to (approximately) sample from the distribution  $\pi(\theta) \propto e^{- \frac{\eps}{2LR}f(\theta, x)}$.
However, if Algorithm \ref{alg_TV_to_pure} has not halted after $t=10\log(\frac{d}{\eps})$ iterations, we stop running Algorithm \ref{alg_TV_to_pure} and instead output $\hat{\theta} = 0 \in K$.

\paragraph{Showing $\eps$-differential privacy.}
Since $f(\theta, x) := \sum_{i=1}^n \ell_i(\theta,x_i)$, where each $\ell_i$ is an $L$-Lipschitz function of $\theta$, we have that $f$ is an  is a $n L$-Lipschitz  function of $\theta$, and hence that $\frac{\eps}{2LR}f$ is a $\frac{n\eps}{R}$-Lipschitz  function of $\theta$.

By Theorem \ref{thm_infinity_divergence_sampler}, Algorithm \ref{alg_TV_to_pure} (with Dikin Walk Markov chain from \cite{narayanan2017efficient} as subroutine), conditional on  Algorithm \ref{alg_TV_to_pure} halting after $t=10\log(\frac{n \eps}{d})$ iterations, 
 outputs a point $\hat{\theta}$ from a distribution $\hat{\nu}$ where
 \begin{equation}\label{eq_w1}
      \mathrm{d}_\infty(\hat{\nu}, \pi) < \eps,
 \end{equation}
 with probability at least 
 \begin{equation}\label{eq_w2}
 \mathbb{P}(\tau \leq t) \geq 1- \left(\frac{2}{3}\right)^{t+1}.     
 \end{equation}
 Otherwise, we output $\hat{\theta} = 0 \in K$.
  To see why the event when our algorithm outputs $\hat{\theta} = 0$ satisfies $\varepsilon$-differential privacy, from Equation \eqref{eq_a6} in the proof of  Theorem \ref{thm_TV_to_inf_divergence} we have that  the probability $\mathbb{P}(\tau = \tau_{\mathrm{max}})$ that Algorithm \ref{alg_TV_to_pure} will reject at all $\tau_{\mathrm{max}}$ iterations satisfies
  
  \begin{equation} \label{eq_revision1}
\left(\frac{1}{2}\right)^{\tau_{\mathrm{max}}} e^{-\frac{\eps}{2}} \leq  \mathbb{P}(\tau = \tau_{\mathrm{max}})  \leq \left(\frac{1}{2}\right)^{\tau_{\mathrm{max}}} e^{\frac{\eps}{2}},\\
\end{equation}
But  $\mathbb{P}(\tau = \tau_{\mathrm{max}})  =  \mathbb{P}(\hat{\theta}= 0)$, since (ignoring events of probability measure zero) Algorithm \ref{alg_TV_to_pure} outputs $\hat{\theta}=0$ if and only if the number of rejections $\tau$ satisfies $\tau = \tau_{\mathrm{max}}$.
Therefore, for any dataset $x \in \mathcal{D}^n$, Inequality \eqref{eq_revision1} implies that the probability $\mathbb{P}(\hat{\theta}= 0) \equiv \mathbb{P}(\hat{\theta}= 0|x)$ that Algorithm \ref{alg_TV_to_pure} outputs the point $0 \in K$ satisfies
  \begin{equation} \label{eq_revision2}
\left(\frac{1}{2}\right)^{\tau_{\mathrm{max}}} e^{-\frac{\eps}{2}} \leq  \mathbb{P}(\hat{\theta}= 0|\, x)  \leq \left(\frac{1}{2}\right)^{\tau_{\mathrm{max}}} e^{\frac{\eps}{2}}.
\end{equation}
 Therefore, for any $x, x' \in \mathcal{D}^n$, Inequality \eqref{eq_revision2} implies that
  $\mathbb{P}(\hat{\theta}= 0|\, x) \leq e^{\eps} \mathbb{P}(\hat{\theta}= 0|\, x')$,
 implying that the event when our Algorithm outputs $\hat{\theta}= 0$ satisfies the definition of $\varepsilon$-differential privacy.
Thus, by Equation \eqref{eq_w1} and Lemma \ref{Lemma_exponential_mechanism},  we have that $\hat{\theta}$ is pure $\eps$-differentially private.

\paragraph{Bounding the ERM utility.}
Moreover, by Equations \eqref{eq_w1}, \eqref{eq_w2} , and Lemma \ref{Lemma_exponential_mechanism}, we have that $\hat{\theta}$ achieves ERM utility
\begin{align*}
    \mathbb{E}[f(\hat{\theta}, x) - \min_{\theta \in K} f(\theta, x)]
    &\leq \mathbb{E}_{\xi \sim \hat{\nu}}[f(\xi, x) - \min_{\theta \in K} f(\theta, x)] + \mathbb{P}(\tau > t) \times [f(0, x) - \min_{\theta \in K} f(\theta, x)]\\
    &\stackrel{\textrm{Eq. }\ref{eq_w1},\ref{eq_w2}}{\leq} e^{\eps}\times \mathbb{E}_{z \sim \pi}[f(z, x) - \min_{\theta \in K} f(\theta, x)] + \left(\frac{2}{3}\right)^{t+1} \times 2nLR\\
    &\leq e^{\eps}\times \mathbb{E}_{z \sim \pi}[f(z, x) - \min_{\theta \in K} f(\theta, x)]+ \left(\frac{2}{3}\right)^{10\log(\frac{n\eps}{d})} \times 2nLR\\
   &\stackrel{\textrm{Lemma }\ref{Lemma_exponential_mechanism}}{\leq}  e^{\eps}\times O\left(\frac{d L R}{\eps}\right) + \frac{d L R}{\eps}\\
   &=O\left(\frac{d L R}{\eps}\right),
\end{align*}
where the second inequality holds since $f$ is $nL$-Lipschitz and $K \subseteq B(0,R)$.

\paragraph{Bounding the number of operations.}
 Moreover, also by Theorem \ref{thm_infinity_divergence_sampler}, the sum $T$ of the number of steps of Dikin Walk \cite{narayanan2017efficient} over all the times it is called by Algorithm \ref{alg_TV_to_pure} is at most $O(t \times (m^2d^3 + m^2d n^2 \eps^2) \times [\eps n + d  \mathrm{log}(\frac{Rd +n \eps d}{r \eps})])$ steps.
 And each step takes $O(md^{\omega-1})$ arithmetic operations, plus one evaluation of the value of $f$ (and hence $n$ evaluations of functions $\ell_i$).
Thus, the number of steps is at most $T=O((m^2d^3 + m^2d n^2 \eps^2) \times [\eps n + d  \mathrm{log}(\frac{Rd +n \eps d}{r \eps})] \times  \log(\frac{n \eps}{d})$ steps, 
where each step takes $O(md^{\omega-1})$ arithmetic operations, plus one evaluation of the function $f$.
Thus, Algorithm \ref{alg_TV_to_pure} finishes in at most   $T \times md^{\omega-1}$ arithmetic operations plus $T$ evaluations of the function $f$,
where $T=O((m^2d^3 + m^2d n^2 \eps^2) \times [\eps n + d]\times \log^2(\frac{n \eps}{d}))$.

\end{proof}

\subsection{Proof of Corollary \ref{corrolary_low_rank}}\label{sec_proof_of_low_rank_DP}

\begin{proof} [{\bf of Corollary \ref{corrolary_low_rank}}]
\cite{leake2020polynomial} show that one can find a pure $\epsilon$-differentially private rank-$k$ projection $P$ such that $E_P[\langle \Sigma, P \rangle] \geq (1-\delta) \sum_{i=1}^k \lambda_i$ whenever  $\sum_{i=1}^k \lambda_i \geq \frac{dk}{\epsilon \delta}  \log \frac{1}{\delta}$ for any $\delta>0$ and some universal constant $C>0$, where $\lambda_1 \geq \cdots \geq \lambda_d >0$ denote the eigenvalues of $\Sigma$, by generating a sample from a linear (and hence log-Lipschitz) log-concave distribution $\pi$ on a polytope $K$ with infinity-distance error $O(\epsilon)$.

Specifically, their linear log-density $\pi$ has Lipschitz constant $L = d^2(\lambda_1 - \lambda_d)$, (first equation in Section 5.2 in the arXiv version of  \cite{leake2020polynomial}).
Their polytope $K$ is in $\mathbb{R}^{\frac{d(d-1)}{2}}$ and is defined by $m=d(d-1)$ inequalities  (Equations (5) and (6) in the arXiv version of \cite{leake2020polynomial}).
Moreover, $K$ is contained in a ball of radius $R= \sqrt{d}$, ((Lemma 4.7) of \cite{leake2020polynomial})
 and contains a ball of radius $r= \frac{1}{8d^2}$ ((Lemma 4.8) of \cite{leake2020polynomial}).

Applying Theorem \ref{thm_infinity_divergence_sampler} with the above parameters for $m,L,r,R$ and the dimension $\frac{d(d-1)}{2}$, we obtain a sample from the distribution $\pi$ with infinity-distance error $O(\epsilon)$ in a number of arithmetic operations that is {\em logarithmic} in $\frac{1}{\epsilon}$ and polynomial in $d$ and $\lambda_1 - \lambda_d$.
\end{proof}

\begin{remark}[Privacy of running time]
In the proof of Theorem \ref{thm_TV_to_inf_divergence} we also show that the probability distribution of the number of iterations $\tau$ of Algorithm \ref{alg_TV_to_pure} satisfies $\left(\frac{1}{2}\right)^t e^{-\frac{\eps}{2}} \leq  \mathbb{P}(\tau = t)  \leq \left(\frac{1}{2}\right)^t e^{\frac{\eps}{2}}$ for all $t \leq \tau_{\mathrm{max}}$.
This ensures that the number of iterations $\tau$ is $\eps$-pure DP; thus a malicious adversary cannot gain much information about the dataset by measuring the time it takes for Algorithm \ref{alg_TV_to_pure} to finish.
\end{remark}

\section{Conclusion, limitations, and future Work}\label{sec:conc}

To the best of our knowledge, this is the first work that presents an algorithm for sampling from logconcave distributions on convex bodies that comes with infinity-distance bounds and whose running time depends logarithmically on $1/\eps$.
Towards this, the main technical contribution is Algorithm \ref{alg_TV_to_pure} (and Theorem \ref{thm_TV_to_inf_divergence}) which achieves this improved dependence on $\eps$ by taking as input  continuous samples from a convex body with TV bounds and converting them to samples with infinity-distance bounds.
On the other hand, our bounds are polynomial in $LR$, yet there are algorithms for sampling from logconcave distributions $\pi \propto e^{-f}$ on a  convex body  in the total variation distance that are poly-logarithmic in $R$ and do not assume $f$ to be Lipschitz \cite{lovasz2006fast}.
Thus, the main open problem that remains is whether one can also obtain running time bounds for sampling in the infinity-distance which are poly-logarithmic in $R$ and do not require $f$ to be Lipschitz.

Our main result also has direct applications to differentially private optimization (Corollaries \ref{corr_DP} and \ref{corrolary_low_rank}).
Differential privacy is a notion which has been embraced in many technologies in societal contexts where privacy of individuals is a concern.
Hence, we see our work to have a potential of positive societal impact and do not foresee any potential negative societal impacts.

\section*{Acknowledgments}
This research was  supported in part by  NSF CCF-2104528,  CCF-1908347, and CCF-2112665 awards.

\newpage
\bibliography{DP}
\bibliographystyle{plain}

\appendix

\section{Performance of Algorithm \ref{alg_TV_to_pure} on Simple Test Functions}

In this section we implement Algorithm 1 on two distributions: a simple one-dimensional distribution, and a $d=100$ dimensional ``Dirichlet'' distribution.
Simulations were performed in Matlab, on a 1.6 GHz Dual-Core Intel Core i5 2019 Macbook Air laptop.

{\em One-dimensional distribution:} We first investigate the performance of Algorithm \ref{alg_TV_to_pure} on a simple one-dimensional distribution $\pi$, to verify that our algorithm generates points within the required infinity distance.
We do this by generating a histogram of the output of Algorithm \ref{alg_TV_to_pure}, by running Algorithm \ref{alg_TV_to_pure} $10^7$ times.
For this experiment, we choose a simple one-dimensional distribution as this allows us to obtain a more precise estimate of the infinity distance, since the number of points needed to compute the histogram grows exponentially with the dimension.

We consider the target distribution $\pi(\theta) \propto e^{-\frac{1}{2}(3-\theta)},$ with support on $K = [-1,3]$.
And we provide Algorithm \ref{alg_TV_to_pure} with samples from a distribution 
$$\mu(\theta) \propto \begin{cases}
e^{-\frac{1}{2}(3-\theta)} \qquad \theta \in K \backslash ([0.499,0.501] \cup [1.999,2.001] \cup [2.999,3])\\
0 \qquad \textrm{otherwise}.
\end{cases} \\
$$
Note that this distribution $\mu$ satisfies $\|\pi - \mu\|_{\mathrm{TV}} \leq \frac{1}{100}$, and yet  $\mathrm{d}_{\infty}(\pi,\mu) = \sup_{\theta \in K} |\log  \frac{\nu(\theta)}{\pi(\theta)}| = \infty$ since there are points $\theta \in K$ where $\nu(\theta) = 0$ but $\pi(\theta) >0$.

We run Algorithm \ref{alg_TV_to_pure} $10^7$ times (with parameters $\varepsilon = 0.1$, $L= \frac{1}{2}$, $R=4$, and $\Delta = \frac{\varepsilon}{\max(d,LR)}= 0.05$), to generate a histogram of the distribution $\nu$ of the output of Algorithm \ref{alg_TV_to_pure} (Figure \ref{fig_simulation}).
We observe that Algorithm \ref{alg_TV_to_pure}  terminates after an average of 2.1904 iterations, and generates points from a distribution $\nu$ with $\mathrm{d}_{\infty}(\nu, \pi) = 0.1054$, roughly matching the value of the parameter $\varepsilon = \frac{1}{10}$.

\begin{figure}[h]
    \centering
    \includegraphics[width=0.7\textwidth]{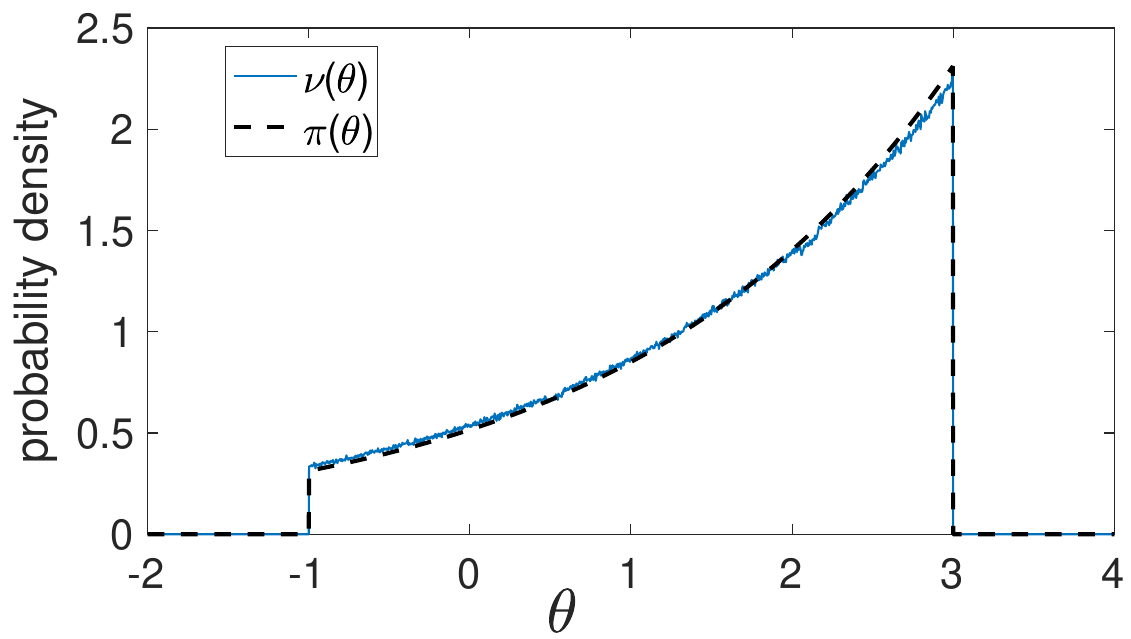}
    \caption{Histogram of the distribution $\nu(\theta)$ of the output of Algorithm \ref{alg_TV_to_pure} (blue curve) for target distribution $\pi(\theta) = e^{-\frac{1}{2}(3-\theta)}$ with support on $K = [-1,3]$ (dashed black curve), when provided with samples from a distribution $\mu$ such that $\|\pi - \mu\|_{\mathrm{TV}} \leq \frac{1}{100}$ and $\mathrm{d}_{\infty}(\pi,\mu) = \infty$.
    Algorithm \ref{alg_TV_to_pure} (with parameter $\varepsilon = 0.1$) terminated after an average of 2.1904 iterations, and generated points from a distribution $\nu$ with $\mathrm{d}_{\infty}(\nu, \pi) = 0.1054$, which roughly matches the choice of parameter $\varepsilon = 0.1$. 
    }
    \label{fig_simulation}
\end{figure}

{\em $100$-dimensional Dirichlet distribution:}  We also implement Algorithm \ref{alg_TV_to_pure}  on a $d=100$ dimensional distribution.
Specifically, we consider the Dirichlet distribution $\pi(\theta) \propto \prod_{i=1}^d \theta_i$ with support on the simplex $K = \{\theta \in \mathbb{R}^d:  \sum_{i=1}^d \theta_i \leq 1, \theta_i \in [0,1] \forall i \in [d] \}$.
And we provide Algorithm \ref{alg_TV_to_pure} with samples from a distribution 
$$\mu(\theta) = \begin{cases}
\pi(\theta) \qquad \theta \notin B(0,\frac{1}{100}) \\
0 \qquad \textrm{otherwise}.
\end{cases} \\
$$
Note that this distribution $\mu$ satisfies $\|\pi - \mu\|_{\mathrm{TV}} < 10^{-d}$, and yet  $\mathrm{d}_{\infty}(\pi,\mu) = \sup_{\theta \in K} |\log  \frac{\nu(\theta)}{\pi(\theta)}| = \infty$ since there are points $\theta \in K$ where $\nu(\theta) = 0$ but $\pi(\theta) >0$.
We observe that Algorithm \ref{alg_TV_to_pure}  terminates after an average of 1.9935 iterations (with the average taken over $10^5$ runs of Algorithm \ref{alg_TV_to_pure}). 
(We do not compute the histogram and infinity distance for the $d=100$ dimensional Dirichlet distribution, since the number of points needed to compute the histogram grows exponentially with $d$.)

\end{document}